\documentclass[twoside,11pt]{article}

\usepackage{amsmath}
\usepackage{amsthm}
\usepackage{amsfonts}
\usepackage{amssymb}
\usepackage{caption}
\usepackage[english]{babel}
\usepackage{fontenc}
\usepackage{graphicx}
\usepackage{mathrsfs}
\usepackage{bbm}
\usepackage{array, multirow}
\usepackage{float}
\usepackage{enumerate}
\usepackage{textcomp}
\usepackage{algorithm}
\usepackage{algorithmic}
\usepackage[font = scriptsize]{subcaption}
\usepackage{xcolor}
\usepackage{hyperref}
\usepackage{calc}

\usepackage{fullpage}

\usepackage{natbib}

\theoremstyle{plain}
\newtheorem{thm}{Theorem}
\newtheorem{lem}[thm]{Lemma}

\newtheorem{cor}[thm]{Corollary}

\theoremstyle{definition}

\newcommand{\mb}{\mathbf}
\newcommand{\mbb}{\boldsymbol}
\newcolumntype{V}{>{\centering\arraybackslash} m{4cm} }
\newcolumntype{M}{>{\centering\arraybackslash} m{0.25cm} }
\newcolumntype{R}[1]{>{\raggedleft\let\newline\\\arraybackslash\hspace{0pt}}m{#1}}
\newcolumntype{L}[1]{>{\raggedright\let\newline\\\arraybackslash\hspace{0pt}}m{#1}}

\author{Rajen D. Shah \\
  Statistical Laboratory, 
  University of Cambridge\\
  r.shah@statslab.cam.ac.uk
}
\title{Modelling Interactions in High-dimensional Data with Backtracking}

\newcommand{\abs}[1]{\left| #1\right|}
\newcommand{\norm}[1]{\left\lVert #1 \right\rVert}
\newcommand{\norms}[1]{\| #1 \|} 

\newcommand{\pr}{\mathbb{P}}
\newcommand{\R}{\mathbb{R}}
\newcommand{\E}{\mathbb{E}}

\newcommand{\Cov}{\mathrm{Cov}}
\newcommand{\Var}{\mathrm{Var}}

\newcommand{\argmin}[1]{\underset{#1}{\operatorname{arg}\operatorname{min}}\;}
\newcommand{\sgn}{\mathrm{sgn}}
\newcommand{\iid}{\stackrel{\mathrm{i.i.d.}}{\sim}}

\author{Rajen D. Shah \\
  Statistical Laboratory\\
  University of Cambridge\\
  r.shah@statslab.cam.ac.uk
}

\title{Modelling Interactions in High-dimensional Data with Backtracking}

\begin{document}
\maketitle
\begin{abstract}%
We study the problem of high-dimensional regression when there may be interacting variables. Approaches using sparsity-inducing penalty functions such as the Lasso can be useful for producing interpretable models. However, when the number variables runs into the thousands, and so even two-way interactions number in the millions, these methods may become computationally infeasible. Typically variable screening based on model fits using only main effects must be performed first. One problem with screening is that important variables may be missed if they are only useful for prediction when certain interaction terms are also present in the model.

To tackle this issue, we introduce a new method we call Backtracking. It can be incorporated into many existing high-dimensional methods based on penalty functions, and works by building increasing sets of candidate interactions iteratively. Models fitted on the main effects and interactions selected early on in this process guide the selection of future interactions. By also making use of previous fits for computation, as well as performing calculations is parallel, the overall run-time of the algorithm can be greatly reduced.

The effectiveness of our method when applied to regression and classification problems is
demonstrated on simulated and real data sets.
In the case of using Backtracking with the Lasso, we also give some theoretical support for our procedure.
\end{abstract}

\section{Introduction} \label{sec:Intro}
In recent years, there has been a lot of progress in the field of high-dimensional regression.
Much of the development has centred around the Lasso \citep{tibshirani96regression}, which given a vector of responses $\mb Y \in \R^n$ and design matrix $\mb{X} \in \R^{n \times p}$, solves
 \begin{equation} \label{eq:Lasso}
 (\hat{\mu}, \hat{\mbb\beta}) :=\argmin{(\mu, \mbb \beta) \in \R \times \R^p} \{ \tfrac{1}{2n}\norms{\mb{Y} - \mu\textbf{1} - \mb{X}\mbb{\beta}}_2 ^2 + \lambda \norms{\mbb \beta}_1\},
\end{equation}
where $\textbf{1}$ is an $n$-vector of ones and the regularisation parameter $\lambda$ controls the relative contribution of the penalty term to the objective.
The many extensions of the Lasso allow most familiar models from classical (low-dimensional) statistics to now be fitted in situations where the number of variables $p$ may be tens of thousands and even greatly exceed the number of observations $n$ (see the monograph \citet{buhlmann2011statistics} and references therein).

However, despite the advances, fitting models with interactions remains a challenge. Two issues that arise are:
\begin{enumerate}[(i)]
\item Since there are $p(p-1)/2$ possible first-order interactions, the main effects can be swamped by the vastly more numerous interaction terms and without proper regularisation, stand little chance of being selected in the final model (see Figure~\ref{fig:all_interact}).
\item Monitoring the coefficients of all the interaction terms quickly becomes infeasible as $p$ runs into the thousands.
\end{enumerate}

\subsection{Related Work}
For situations where $p < 1000$ or thereabouts and the case of two-way interactions, a lot of work has been done in recent years to address this need. To tackle (i), many of the proposals use penalty functions and constraints designed to enforce that if an interaction term is in the fitted model, one or both main effects are also present \citep{LinZhang2006, Zhaoetal2009, Yuanetal2009, RadchenkoJames2010, Jenatton_etal2011, Bach_etal2012, Bach_etal2012a, Bien_etal2013, Lim2015, Haris2015}. See also \citet{Turlach2004} and \citet{Yuanetal2007}, which consider modifications of the LAR algorithm \citet{efron04least} that impose this type of condition.

In the moderate-dimensional setting that these methods are designed for, the computational issue (ii) is just about manageable. However, when $p$ is larger---the situation of interest in this paper---it typically becomes necessary to narrow the search for interactions. Comparatively little work has been done on fitting models with interactions to data of this sort of dimension. An exception is the method of Random Intersection Trees \citep{Shah2014}, which does not explicitly restrict the search space of interactions. However this is designed for a classification setting with a binary predictor matrix and does not fit a model but rather tries to find interactions that are marginally informative.

One option is to screen for important variables and only consider interactions involving the selected set. \citet{Wu_etal2010} and others take this approach: the Lasso is first used to select main effects; then interactions between the selected main effects are added to the design matrix, and the Lasso is run once more to give the final model.

The success of this method relies on all main effects involved in interactions being selected in the initial screening stage. However, this may well not happen. Certain interactions may need to be included in the model before some main effects can be selected. To address this issue,
 \citet{Bickel_etal2010} propose a procedure involving sequential Lasso fits which, for some predefined number $K$, selects $K$ variables from each fit and then adds all interactions between those variables as candidate variables for the following fit. The process continues until all interactions to be added are already present. However, it is not clear how one should choose $K$: a large $K$ may result in a large number of spurious interactions being added at each stage, whereas a small $K$ could cause the procedure to terminate before it has had a chance to include important interactions.
 

Rather than adding interactions in one or more distinct stages, when variables are selected in a greedy fashion, the set of candidate interactions can be updated after each selection. This dynamic updating of interactions available for selection is present in the popular MARS procedure of \citet{MARS}. One potential problem with this approach is that particularly in high-dimensional situations, overly greedy selection can sometimes produce unstable final models and predictive performance can suffer as a consequence.

The iFORT method of \citet{Hao2014} applies forward selection to a dynamically updated set of candidate interactions and main effects, for the purposes of variable screening. In this work, 
we propose a new method we call Backtracking, for incorporating a similar model building strategy to that of MARS and iFORT into methods based on sparsity-inducing penalty functions. Though greedy forward selection methods often work well, penalty function-based methods such as the Lasso can be more stable (see \citet{efron04least}) and offer a useful alternative. 


\subsection{Outline of the Idea}
When used with the Lasso, Backtracking begins by computing the Lasso solution path, decreasing $\lambda$ from $\infty$. A second solution path, $P_2$, is then produced, where the design matrix contains all main effects, and also the interaction between the first two active variables in the initial path. Continuing iteratively, subsequent solution paths $P_3,\ldots, P_T$ are computed where the set of main effects and interactions in the design matrix for the $k$th path is determined based on the previous path $P_{k-1}$. Thus if in the third path, a key interaction was included and so variable selection was then more accurate, the selection of interactions for 
all 
future paths would benefit. In this way information is used as soon as it is available, rather than at discrete stages as with the method of \citet{Bickel_etal2010}. In addition, if all important interactions have already been included by  $P_3$, we have a solution path unhindered by the addition of further spurious interactions. 

It may seem that a drawback of our proposed approach is that the computational cost of producing all $T$ solution paths will usually be unacceptably large. However, computation of the full collection of solution paths is typically very fast. This is because rather than computing each of the solution paths from scratch, for each new solution path $P_{k+1}$, we first track along the previous path $P_k$ to find where $P_{k+1}$ departs from $P_k$. This is the origin of the name Backtracking. Typically, checking whether a given trial solution is on a solution path requires much less computation than calculating the solution path itself, and so this Backtracking step is rather quick. Furthermore, when the solution paths do separate, the tail portions of the paths can be computed in parallel.

An R \citep{R} package for the method is available on the author's website.

\subsection{Organisation of the Paper}
The rest of the paper is organised as follows. In Section~\ref{sec:motivation} we describe an example which provides some motivation for our Backtracking method. In Section~\ref{sec:Backtrack_lasso} we develop our method in the context of the Lasso for the linear model. In Section~\ref{sec:Extensions}, we describe how our method can be extended beyond the case of the Lasso for the linear model. In  Section~\ref{sec:numerical} we report the results of some simulation experiments and real data analyses that demonstrate the effectiveness of Backtracking. Finally, in Section~\ref{sec:theory}, we present some theoretical results which aim to give a deeper understanding of the way in which Backtracking works. Proofs are collected in the appendix.

\section{Motivation} \label{sec:motivation}
In this section we introduce a toy example where approaches that select candidate interactions based on selected main effects will tend to perform poorly. We consider a linear model with interactions involving a design matrix $\mb X \in \R^{n \times p}$ with $n=200$, $p=500$ and where
\begin{equation} \label{eq:motivation}
 Y_i = \sum_{j = 1} ^6 \beta_j X_{ij} + \beta_7 X_{i1} X_{i2} + \beta_8 X_{i3} X_{i4} + \beta_9 X_{i5} X_{i6} + \varepsilon_i, \quad \varepsilon_i \sim N(0, \sigma^2), \quad i = 1, \ldots, n.
\end{equation}
We take $\mb X$ with i.i.d.\ rows having a distribution such that
$X_{i5}$ is uncorrelated with $\{X_{ij}: j \neq 5\}$. We then choose $\beta_1, \ldots, \beta_9$ in such a way that $X_{i5}$ is also uncorrelated with the response yet $\beta_5 \neq 0$. The precise construction is detailed in the appendix.

In order to select variable 5 using that Lasso, we would need to have already selected some important interactions.
Thus if we first select important main effects using the Lasso, for example, it is very unlikely that variable 5 will be selected. Then if we add all two-way interactions between the selected variables and fit the Lasso once more, the interaction between variables 5 and 6 will not be included. Of course, one can again add interactions between selected variables and compute another Lasso fit, and then there is a chance the interaction will be selected. Thus it is very likely that at least three Lasso fits will be needed in order to select the right variables.

Figure~\ref{fig:main_eff} shows the result of applying the Lasso to data generated according to \eqref{eq:motivation}, $\sigma$ chosen to give a signal-to-noise ratio (SNR) of 4, and
\[
 \mbb \beta = (-1.25, -0.75, 0.75, -0.5, -2, 1.5, 2, 2, 1)^T.
\]
As expected, we see variable 5 is nowhere to be seen and instead many unwanted variables are selected as $\lambda$ is decreased. Figure~\ref{fig:all_interact} illustrates the effect of including all $p(p-1) / 2$ possible interactions in the design matrix. Even in our rather moderate-dimensional situation, we are not able to recover the true signal. Though all the true interaction terms are selected, now neither variable 4 nor variable 5 are present in the solution paths and many false interactions are selected. 

Although this example is rather contrived, it illustrates how sometimes the right interactions need to be augmented to the design matrix in order for certain variables to be selected. Even when interactions are only present if the corresponding main effects are too, main effects can be missed by a procedure that does not consider interactions. In fact, we can see the same phenomenon occurring when the design matrix has i.i.d.\ Gaussian entries (see Section~\ref{sec:simulations}).
Thus multiple Lasso fits might be needed to have any chance of selecting the right model.

This raises the question of which tuning parameters to use in the multiple Lasso fits. One option, which we shall refer to as the iterated Lasso, is to select tuning parameters by cross-validation each time. A drawback of this approach, though, is that the number of interactions to add can be quite large if cross-validation chooses a large active set. This is often the case when the presence of interactions makes some important main effects hard to distinguish from noise variables in the initial Lasso fit. Then cross-validation may choose a low $\lambda$ in order to try to select those variables, but this would result in many noise variables also being included in the active set.

We take an alternative approach here and include suspected interactions in the design matrix as soon as possible. That is, if we progress along the solution path from $\lambda = \infty$, and two variables enter the model, we immediately add their interaction to the design matrix and start computing the Lasso again. We could now disregard the original path, but there is little to lose, and possibly much to gain, in continuing the original path in parallel with the new one.
We can then repeat this process, adding new 
interactions when necessary, and restarting the Lasso, whilst still continuing all previous paths in parallel.
We show in the next section how computation can be made very fast since many of these solution paths will share the same initial portions.

\begin{figure}
        \begin{subfigure}[b]{0.45\textwidth}
                \centering
                \includegraphics[width=\textwidth]{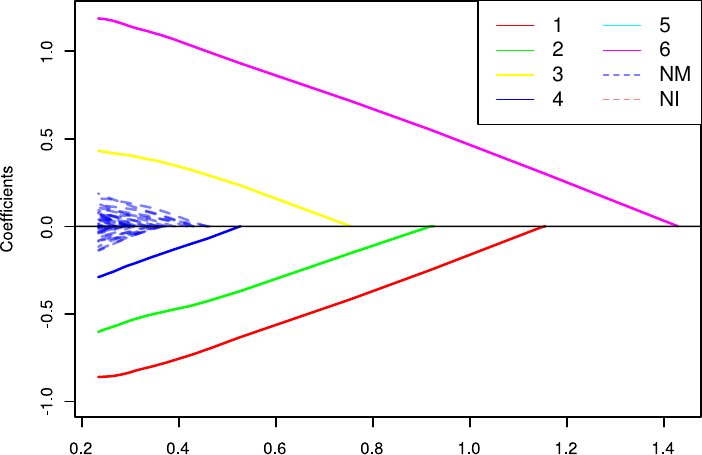}
                \caption{Main effects only}
                \label{fig:main_eff}
        \end{subfigure}
\qquad
        \begin{subfigure}[b]{0.45\textwidth}
                \centering
                \includegraphics[width=\textwidth]{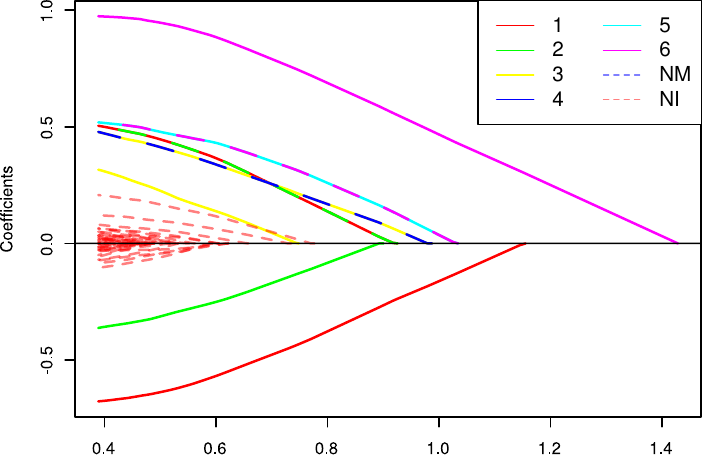}
                \caption{All interactions added}
                \label{fig:all_interact}
        \end{subfigure}

\vspace{14pt}
	\begin{subfigure}[b]{0.45\textwidth}
                \centering
                \includegraphics[width=\textwidth]{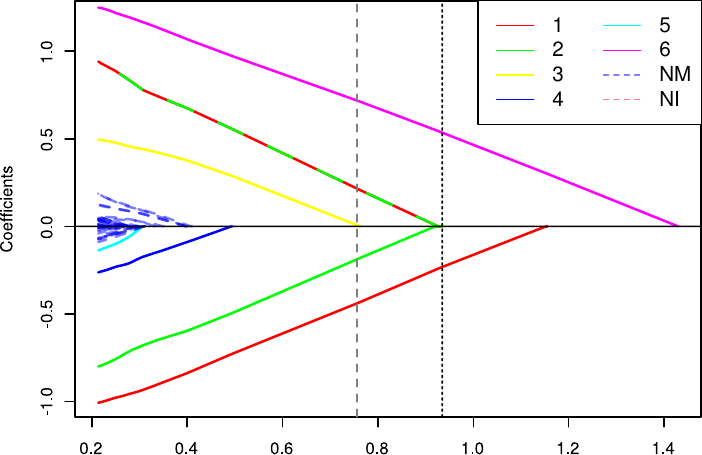}
                \caption{Step 3: $\{1,2\}$, $\{2,6\}$. $\{1,6\}$ added in step 2.}
                \label{fig:step_3}
        \end{subfigure}
\qquad
	\begin{subfigure}[b]{0.45\textwidth}
                \centering
                \includegraphics[width=\textwidth]{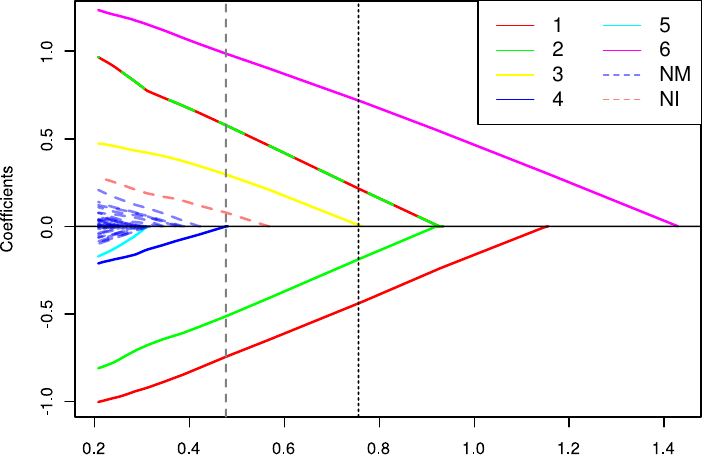}
                \caption{Step 4: $\{1,3\}$, $\{2,3\}$, $\{3,6\}$.}
                \label{fig:step_4}
        \end{subfigure}

\vspace{14pt}
	\begin{subfigure}[b]{0.45\textwidth}
                \centering
                \includegraphics[width=\textwidth]{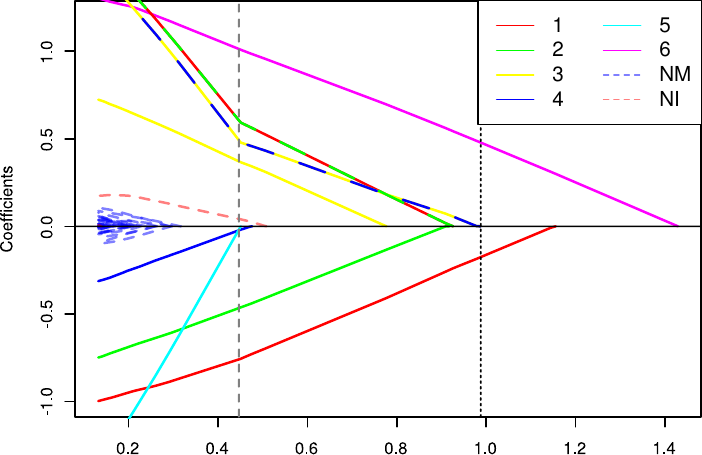}
                \caption{Step 5: $\{1,4\}$, $\{2,4\}$, $\{3,4\}$, $\{4,6\}$.}
                \label{fig:step_5}
        \end{subfigure}
\qquad
	\begin{subfigure}[b]{0.45\textwidth}
                \centering
                \includegraphics[width=\textwidth]{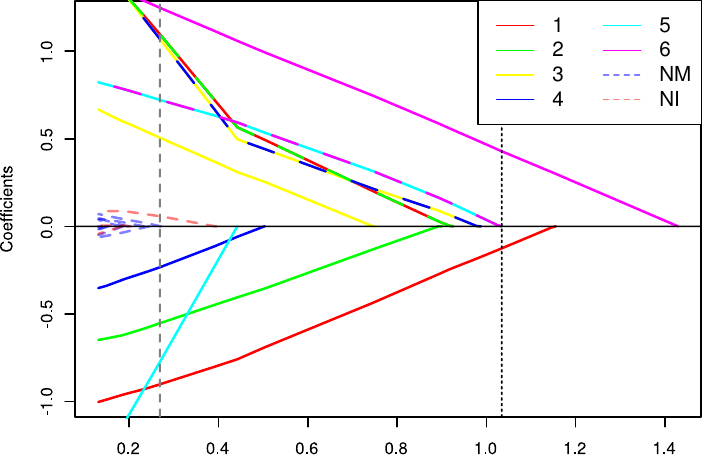}
                \caption{Step 6: $\{1,5\}$, $\{2,5\}$, $\{3,5\}$, $\{4,5\}$, $\{5,6\}$.}
                \label{fig:step_6}
        \end{subfigure}
   \caption{For data generated as described in Section~\ref{sec:motivation}, the coefficient paths against $\lambda$ of the Lasso with main effects only, (a); the Lasso with all interactions added, (b); and Backtracking with $k=3, \ldots, 6$, ((c)--(d)); when applied to the example in Section~\ref{sec:motivation}. Below the Backtracking solution paths we give $C_k \setminus C_{k-1}$: the interactions which have been added in the current step. The solid red, green, yellow, blue, cyan and magenta lines trace the coefficients of variables $1, \ldots, 6$ respectively, with the alternately coloured lines representing the corresponding interactions. The dotted blue and red coefficient paths indicate noise main effect (`NM') and interaction (`NI') terms respectively. Vertical dotted black and dashed grey lines give the values of $\lambda^\mathrm{start} _k$ and $\lambda^\mathrm{add} _k$ respectively.}\label{fig:lasso_plots}
\end{figure}

\section{Backtracking with the Lasso} \label{sec:Backtrack_lasso}
In this section we introduce a version of the Backtracking algorithm applied to the Lasso
\eqref{eq:Lasso}.
First, we present a naive version of the algorithm, which is easy to understand. Later in Section~\ref{sec:speedup}, we show that this algorithm performs a large number of unnecessary calculations, and we give a far more efficient version.
\subsection{A Naive Algorithm}
As well as a base regression procedure, the other key ingredient that Backtracking requires is a way of suggesting candidate interactions based on selected main effects, or more generally a way of suggesting higher-order interactions based on lower-order interactions. In order to discuss this and present our algorithm, we first introduce some notation concerning interactions.

Let $\mb X$ be the original $n \times p$ design matrix, with no interactions. In order to consider interactions in our models, rather than indexing variables by a single number $j$, we use subsets of $\{1, \ldots, p\}$. Thus by variable $\{1, 2\}$, we mean the interaction between variables 1 and 2, or in our new notation, variables $\{1\}$ and $\{2\}$. When referring to main effects $\{j\}$ however, we will often omit the braces.
As we are using the Lasso as the base regression procedure here, interaction $\{1, 2\}$ will be the componentwise product of the first two columns of $\mb X$. We will write $\mb X_v \in \R^n$ for variable $v$.

The choice of whether and how to scale and centre interactions and main effects can be a rather delicate one, where domain knowledge may play a key role. In this work, we will centre all main effects, and scale them to have $\ell_2$-norm $\sqrt{n}$. The interactions will be created using these centred and scaled main effects, and they themselves will also be centred and scaled to have $\ell_2$-norm $\sqrt{n}$.

For $C$ a set of subsets of $\{1,\ldots,p\}$ we can form a modified design matrix $\mb X_C$, where the columns of $\mb X_C$ are given by the variables in $C$, centred and scaled as described above. Thus $C$ is the set of candidate variables available for selection when design matrix $\mb X_C$ is used. This subsetting operation will always be taken to have been performed before any further operations on the matrix, so in particular $\mb X_C ^T$ means $(\mb X_C)^T$.

We will consider 
all associated vectors and matrices as indexed by variables, so we may speak of component $\{1, 2\}$ of $\mbb\beta$, denoted $\beta_{\{1, 2\}}$, if $\mbb \beta$ were multiplying a design matrix which included $\{1, 2\}$. Further, for any collection of variables $A$, we will write $\mbb \beta_A$ for the subvector whose components are those indexed by $A$.
To represent an arbitrary variable which may be an interaction, we shall often use $v$ or $u$ and reserve $j$ to index main effects.

We will often need to express the dependence of the Lasso solution $\hat{\mbb \beta}$ \eqref{eq:Lasso} on the tuning parameter $\lambda$ and the design matrix used. We shall write $\hat{\mbb \beta}(\lambda, C)$ when $\mb X_C$ is the design matrix. We will denote the set of active components of a solution $\hat{\mbb \beta}$ by $\mathcal{A}(\hat{\mbb \beta}) = \{v: \hat{\beta}_v \neq 0\}$.

We now introduce a function $\mathcal{I}$ that given a set of variables $A$, suggests a set of interactions to add to the design matrix.
The choice of $\mathcal{I}$ we use here is as follows:
\[
 \mathcal{I}(A) = \{v \subseteq \{1, \ldots, p\}: \text{ for all } u \subsetneq v, \, u \neq \emptyset, \, u \in A\}.
\]
In other words, $\mathcal{I}(A)$ is the set of variables not in $A$, all of whose corresponding lower order interactions are present in $A$. To ease notation, when $A$ contains only main effects $j_1,\ldots,j_s$, we will write  $\mathcal{I}(j_1,\ldots,j_s)=\mathcal{I}(A)$.
For example, $\mathcal{I}(1,2) = \{\{1, 2\}\}$, and $\mathcal{I}(1,2,3) = \{\{1, 2\}, \{2, 3\}, \{1, 3\}\}$. Note $\{1, 2, 3\} \notin \mathcal{I}(1,2,3)$ as the lower order interaction $\{1, 2\}$ of $\{1, 2, 3\}$ is not in $\{\{1\}, \{2\}, \{3\}\}$, for example. Other choices for $\mathcal{I}$ can be made, and we discuss some further possibilities in Section~\ref{sec:Extensions}.

Backtracking relies on a path algorithm for computing the Lasso on a grid of $\lambda$ values $\lambda_1 > \cdots \lambda_L$. Several algorithms are available and coordinate descent methods \citep{friedman2010regularization} appears to work well in practice.


We are now in a position to introduce a naive version of our Backtracking algorithm applied to the Lasso (Algorithm~\ref{alg:naive}). We will assume that the response $\mb Y$ is centred in addition to the design matrix, so no intercept term is necessary. 

\begin{algorithm}
\caption{A naive version of Backtracking with the Lasso} \label{alg:naive}
Set $T$ to be the (given) maximum number of candidate interaction sets to generate.
Let the initial candidate set consist of just main effects: $C_1=\{\{1\},\ldots,\{p\}\}$. Set the index for the candidate sets $k=1$. Let $\lambda_1^\text{start}=\lambda_1$, the largest $\lambda$ value on the grid.
In the steps which follow, we maintain a record of the set of variables which have been non-zero at any point in the algorithm up to the current point (an ``ever active set'', $A$).
\begin{enumerate}
\item Compute the solution path of the Lasso with candidate set $C_k$ from $\lambda_k^\text{start}$ onwards until the ever active set $A$ has $\mathcal{I}(A) \nsubseteq C_k$ (if the smallest $\lambda$ value on the grid is reached then go to 5). Let the $\lambda$ value where this occurs be $\lambda_k^\text{add}$. We will refer to this solution path as $P_k$.
\item Set $C_{k+1} = C_k \cup \mathcal{I}(A)$ so the next candidate set contains all interactions between variables in the ever active set.
\item Set $\lambda_{k+1}^\text{start}=\lambda_1$.
\item Increment $k$. If $k > T$ go to 5, otherwise go back to 1.
\item For each $k$ complete the solution path $P_k$ by continuing it until $\lambda = \lambda_L$. Computing these final pieces of the solution paths can be done in parallel.
\end{enumerate}
\end{algorithm}

The algorithm computes Lasso solution paths whose corresponding design matrices include interactions chosen based on previous paths. 
The quantity $\lambda^{\mathrm{add}}_k$ records the value of $\lambda$ at which interaction terms were added to the set of candidates $C_k$. Here $\lambda_k^\text{start}$ is a redundant quantity and can be replaced everywhere with $\lambda_1$ to give the same algorithm. We include it at this stage though to aid with the presentation of an improved version of the algorithm where $\lambda_k^\text{start}$ in general takes values other than $\lambda_1$.
We note that the final step of completing the solution paths can be carried out as the initial paths are being created, rather than once all initial paths have been created.
Though here the algorithm can include three-way or even higher order interactions, it is straightforward to restrict the possible interactions to be added to first-order interactions, for example.

\subsection{An Improved Algorithm} \label{sec:speedup}
The process of performing multiple Lasso fits is computationally cumbersome, and an immediate gain in efficiency can be realised by noticing that the final collection of solution paths is in fact a tree of solutions: many of the solution paths computed will share the same initial portions.

To discuss this, we first recall the KKT conditions for the Lasso dictate that $\hat{\mbb \beta}$ is a solution to \eqref{eq:Lasso} when the design matrix is $\mb X_C$ if and only if
\begin{align}
 \tfrac{1}{n}\mb X_v ^T (\mb Y - \mb X_C \hat{\mbb \beta}) & = \lambda \sgn(\hat{\beta}_v ) \quad \text{for } \hat{\beta}_v \neq 0 \label{eq:KKT_non_zero} \\
\tfrac{1}{n}|\mb X_v ^T (\mb Y - \mb X_C \hat{\mbb \beta})| & \leq \lambda \quad \text{for } \hat{\beta}_v = 0. \label{eq:KKT_zero}
\end{align}
Note the $ \hat{\mu} \mb X_v ^T \textbf{1}$ term vanishes as the columns of $\mb X_C$ are centred.

We see that if for some $\lambda$
\begin{equation} \label{eq:KKT_check_Lasso}
 \tfrac{1}{n}\norms{\mb X_{C_{k+1} \setminus C_k} ^T (\mb Y - \mb X_{C_k} \hat{\mbb \beta}(\lambda, C_k))}_\infty \leq \lambda,
\end{equation}
then
\begin{align*}
 \hat{\mbb \beta}_{C_{k+1} \setminus C_k} (\lambda, C_{k+1}) = \mb 0, \qquad
\hat{\mbb \beta}_{C_k} (\lambda, C_{k+1}) = \hat{\mbb \beta}(\lambda, C_k).
\end{align*}
Thus given solution path $P_k$, we can attempt to find the smallest $\lambda_l$ such that \eqref{eq:KKT_check_Lasso} holds. Up to that point then,  path $P_{k+1}$ will coincide with $P_k$ and so those Lasso solutions need not be re-computed. Note that verifying \eqref{eq:KKT_check_Lasso} is a computationally simple task requiring only $O(|C_{k+1}\setminus C_k| n)$ operations.

Our final Backtracking algorithm therefore replaces step 3 of Algorithm~\ref{alg:naive} with the following:
\begin{enumerate}
\item[3a.] Find the smallest $\lambda_1 \geq \lambda_l \geq \lambda_k^\text{add}$ such that \eqref{eq:KKT_check_Lasso} holds with $\lambda=\lambda_l$ and set this to be $\lambda_{k+1}^\text{start}$. If no such $\lambda_l$ exists, set $\lambda_{k+1}^\text{start}$ to be $\lambda_1$. 
\end{enumerate}

Figures~\ref{fig:step_3}--\ref{fig:step_6} show steps 3--6 (i.e. $k=3,\ldots,6$) of Backtracking applied to the example described in Section~\ref{sec:motivation}. Note that Figure~\ref{fig:main_eff} is in fact step 1. Step 2 is not shown as the plot looks identical to that in Figure~\ref{fig:main_eff}. We see that when $k=6$, we have a solution path where all the true variable and interaction terms are active before any noise variables enter the coefficient plots.

We can further speed up the algorithm by first checking if $P_k$ coincides with $P_{k+1}$ at $\lambda_k^\text{add}$. If not, we can perform a bisection search to find any point where $P_k$ and $P_{k+1}$ agree, but after which they disagree. This avoids checking \eqref{eq:KKT_check_Lasso} for every $\lambda_l$ up to $\lambda_k^\text{add}$. We will work with the simpler version of Backtracking here using step 3a, but use this faster version in our implementation.

\section{Further Applications of Backtracking} \label{sec:Extensions}
Our Backtracking algorithm has been presented in the context of the Lasso for the linear model. However, the real power of the idea is that it can be incorporated into any method that produces a path of increasingly complex sparse solutions by solving a family of convex optimisation problems parametrised by a tuning parameter. For the Backtracking step, the KKT conditions for these optimisation problems provide a way of checking whether a given trial solution is an optimum. As in the case of the Lasso, checking whether the KKT conditions are satisfied typically requires much less computational effort than computing a solution from scratch. Below we briefly sketch some applications of Backtracking to a few of the many possible methods with which it can be used.

\subsection{Multinomial Regression} \label{sec:multinomial}
An example, which we apply to real data in Section~\ref{sec:real_data}, is multinomial regression with a group Lasso \citep{yuan2006model} penalty. Consider $n$ observations of a categorical response that takes $J$ levels, and $p$ associated covariates. Let $\mb Y$ be the indicator response matrix, with $ij$th entry equal to 1 if the $i$th observation takes the $j$th level, and 0 otherwise. We model
\begin{equation*}
 \pr(Y_{ij} = 1) := \Pi_{ij} (\mbb \mu^*, \mbb \beta^*; \mb X_{S^*}) := \frac{\exp \left(\mu^* _j  + \left( \mb X_{S^*} \mbb \beta^* _j \right)_i \right)}{\sum_{j' = 1} ^J \exp \left(\mu^* _{j'}  + \left( \mb X_{S^*} \mbb \beta^* _{j'} \right)_i \right)}.
\end{equation*}
 Here $\mb \mu^*$ is a vector of intercept terms and  $\mbb \beta^*$ is a $|S^*| \times J$ matrix of coefficients; $\mbb \beta^* _j$ denotes the $j$th column of $\mbb \beta^*$. This model is over-parametrised, but regularisation still allows us produce estimates of $\mbb\mu^*$ and $\mbb\beta^*$ and hence also of $\mbb\Pi$ (see \citet{friedman2010regularization}). When our design matrix is $\mb X_C$, these estimates are given by $(\hat{\mbb\mu}, \hat{\mbb \beta}) := \argmin{\mbb\mu, \mbb\beta} Q(\mbb\mu, \mbb\beta; \lambda)$ where
\begin{multline*}
 Q(\mbb\mu, \mbb\beta; \lambda) := \tfrac{1}{n} \sum_{j = 1} ^J \mb Y_j ^T (\mu_j \textbf{1} + \mb X_C \mbb\beta_j) - 
 \tfrac{1}{n} \textbf{1}^T \log \left( \sum_{j = 1} ^J \exp( \mu_j \textbf{1} + \mb X_C \mbb\beta_j \right) + \lambda \sum_{v \in C} \norms{(\mbb\beta^T)_v}_2.
\end{multline*}
The functions $\log$ and $\exp$ are to be understood as applied componentwise and the rows of $\mbb\beta$ are indexed by elements of $C$. To derive the Backtracking step for this situation, we turn to the KKT conditions which characterise the minima of $Q$:
\begin{align*}
\tfrac{1}{n} \{\mb Y^T - \mbb\Pi^T (\hat{\mbb\mu}, \hat{\mbb \beta}; \mb X_C) \} \textbf{1} & = \mb 0, \\
 \tfrac{1}{n} \{\mb Y^T - \mbb\Pi^T (\hat{\mbb\mu}, \hat{\mbb \beta}; \mb X_C) \} \mb X_v & = -\lambda \frac{({\hat{\mbb \beta}^T})_v}{\norms{({\hat{\mbb \beta}^T})_v}_2} \quad \text{for } ({\hat{\mbb \beta}^T})_v \neq \mb 0, \\
\tfrac{1}{n} \norms{\{\mb Y^T - \mbb \Pi^T (\hat{\mbb \mu}, \hat{\mbb \beta}; \mb X_C) \} \mb X_v}_2 & \leq \lambda \quad \text{for } ({\hat{\mbb \beta}^T})_v = \mb 0.
\end{align*}
Thus, analogously to \eqref{eq:KKT_check_Lasso}, for $D \supsetneq C$, $(\hat{\mbb \beta} ^T (\lambda, D))_{D \setminus C} = \mb 0$ and $(\hat{\mbb \beta}^T (\lambda, D))_C = \hat{\mbb \beta} ^T (\lambda, C)$ if and only if
\begin{equation*}
 \max_{v \in D \setminus C} \tfrac{1}{n} \norms{\{ \mb Y^T - \mbb \Pi^T (\hat{\mbb\mu}(\lambda, C), \hat{\mbb \beta}(\lambda, C); \mb X_{C}) \} \mb X_v}_2 \leq \lambda.
\end{equation*}

\subsection{Structural Sparsity}
Although in our Backtracking algorithm, interaction terms are only added as candidates for selection when all their lower order interactions and main effects are active, this hierarchy in the selection of candidates does not necessarily follow through to the final model: one can have first-order interactions present in the final model without one or more of their main effects, for example. One way to enforce the hierarchy constraint in the final model is to use a base procedure which obeys the constraint itself. Examples of such base procedures are provided by the Composite Absolute Penalties (CAP) family \citep{Zhaoetal2009}.

Consider the linear regression setup 
with interactions. For simplicity we only describe Backtracking with first-order interactions. Let $C$ be the candidate set and let $I = C \setminus C_1$ be the (first-order) interaction terms in $C$. In order to present the penalty, we borrow some notation from Combinatorics. Let $C_1 ^{(r)}$ denote the set of $r$-subsets of $C_1$. For $A \subseteq C_1 ^{(r)}$ and $r \geq 1$, define
\begin{align*}
 \partial_l (A) & = \{v \in C_1 ^{(r-1)} : v \subset u \text{ for some } u \in A\} \\
\partial_u (A) & = \{v \in C_1 ^{(r+1)} : v \subset u \text{ for some } u \in A\}
\end{align*}
These are known as the \emph{lower shadow} and \emph{upper shadow} respectively \citep{Bollobas}.

Our objective function $Q$ is given by
\begin{equation*}
 Q(\mu, \mbb\beta) = \tfrac{1}{2n} \norms{\mb Y - \mu \textbf{1} - \mb X_C\mbb\beta}_2 ^2 + \lambda \norms{\mbb\beta_{C_1 \setminus \partial_l (I)}}_1 + \lambda \sum_{v \in \partial_l (I)} \norms{\mbb\beta_{\{v\} \cup (\partial_u (\{v\}) \cap I)}}_\gamma + \lambda \norms{\mbb\beta_I}_1,
\end{equation*}
where $\gamma > 1$. For example, if $C = \{\{1\}, \ldots, \{4\}, \{1, 2\}, \{2, 3\}\}$, then omitting the factor of $\lambda$, the penalty terms in $Q$ are
\[
 |\beta_{4}| + \norms{(\beta_{1}, \beta_{\{1, 2\}})^T}_\gamma + \norms{(\beta_{2}, \beta_{\{1, 2\}}, \beta_{\{2, 3\}})^T}_\gamma + \norms{(\beta_{3}, \beta_{\{2, 3\}})^T}_\gamma + |\beta_{\{1, 2\}}| + |\beta_{\{2, 3\}}|.
\]
The form of this penalty forces interactions to enter the active set only after or with their corresponding main effects.

The KKT conditions for this optimisation take a more complicated form than those for the Lasso. Nevertheless, checking they hold for a trial solution is an easier task than computing a solution.

\subsection{Nonlinear Models}
If a high-dimensional additive modelling method \citep{Ravikumaretal2009, Meier2009} is used as the base procedure, it is possible to fit nonlinear models with interactions. Here each variable is a collection of basis functions, and to add an interaction between variables, one adds the tensor product of the two collections of basis functions, penalizing the new interaction basis functions appropriately. Structural sparsity approaches can also be used here. The VANISH method of \citet{RadchenkoJames2010} uses a CAP-type penalty in nonlinear regression, and this can be used as a base procedure in a similar way to that sketched above.

\subsection{Introducing more Candidates}
In our description of the Backtracking algorithm, we only introduce an interaction term when \emph{all} of its lower order interactions and main effects are active. Another possibility, in the spirit of MARS \citep{MARS}, is to add interaction terms when \emph{any} of their lower order interactions or main effects are active. As at the $k$th step of Backtracking, there will be roughly $kp$ extra candidates, an approach that can enforce the hierarchical constraint may be necessary to allow main effects to be selected from amongst the more numerous interaction candidates. The key point to note is that if the algorithm is terminated after $T$ steps, we are having to deal with roughly at most $Tp$ variables rather than $O(p^2)$, the latter coming from including all first-order interactions.

Another option proposed by a referee is to augment the initial set of candidates with interactions selected through a simple marginal screening step. If only pairwise interactions are considered here, then this would require $O(p^2 n)$ operations. Though this would be infeasible for very large $p$, for moderate $p$ this would allow important interactions whose corresponding main effects are not strong to be selected.

\section{Numerical Results} \label{sec:numerical}
In this section we evaluate the performance of Backtracking on both simulated and real data sets.
\subsection{Simulations} \label{sec:simulations}
Here we consider five numerical studies designed to demonstrate the effectiveness of Backtracking with the Lasso and also highlight some of the drawbacks of using the Lasso with main effects only, when interactions are present. In each of the five scenarios, we generated 200 design matrices with $n=250$ observations and $p=1000$ covariates. The rows of the design matrices were sampled independently from $N_p (\mb 0, \mbb\Sigma)$ distributions. The covariance matrix $\mbb\Sigma$ was chosen to be the identity in all scenarios except scenario 2, where
\[
 \Sigma_{ij} = 0.75 ^{-||i-j| - p/2| + p/2}.
\]
Thus in this case, the correlation between the components decays exponentially with the distance between them in $\mathbb{Z} / p \mathbb{Z}$.

We created the responses according to the linear model with interactions and set the intercept to 0:
\begin{equation}
 \mb Y = \mb X_{S^*} \mbb\beta_{S^*} ^* + \mbb\varepsilon, \quad \varepsilon_i \iid N(0, \sigma^2).
\end{equation}
The error variance $\sigma^2$ was chosen to achieve a signal-to-noise ratio (SNR) of either 2 or 3.
The set of main effects in $S^*$, $S^* _1$, was $1, \ldots, 10$. The subset of variables involved in interactions was $1,\ldots,6$. The set of first-order interactions in $S^*$ chosen in the different scenarios, $S^* _2$, is displayed in Table~\ref{tab:sim_settings}, and we took $S^* = S^* _1 \cup S^* _2$ so $S^*$ contained no higher order interactions. In each simulation run, $\mbb \beta^* _{S^* _1}$ was fixed and given by
\[
 (2, -1.5, 1.25, -1, 1, -1, 1, 1, 1, 1)^T .
\]
Each component of $\mbb \beta^* _{S^* _2}$ was chosen to be $\sqrt{\norms{\mbb \beta^* _{S^* _1}}_2 ^2 / \abs{S^* _1}}$. Thus the squared magnitude of the interactions was equal to average of the squared magnitudes of the main effects.

\begin{table}
\centering
\begin{tabular}{cc}
\hline
Scenario & $S^* _2$\\
\hline
1 & $\emptyset$\\
2 & $\emptyset$\\
3 & $\{\{1,2\}, \{3,4\}, \{5,6\}\}$\\
4 & $\{\{1,2\}, \{1,3\},\ldots,\{1,6\}\}$\\
5 & $\mathcal{I}(1, 2, 3) \cup \mathcal{I}(4, 5, 6)$\\
\hline
\end{tabular}
\caption{Simulation settings.}
\label{tab:sim_settings}
\end{table}

In all of the scenarios, we applied four methods: the Lasso using only the main effects; iterated Lasso fits; marginal screening for interactions followed by the Lasso; and the Lasso with Backtracking.
Note that due to the size of $p$ in these examples, most of the methods for finding interactions in lower-dimensional data discussed in Section~\ref{sec:Intro}, are computationally impractical here.

For the iterated Lasso fits, we repeated the following process. Given a design matrix, first fit the Lasso. Then apply 5-fold cross-validation to give a $\lambda$ value and associated active set. Finally add all interactions between variables in this active set to the design matrix, ready for the next iteration. For computational feasibility, the procedure was terminated when the number of variables in the design matrix exceeded $p + 250 \times 249 / 2$.

With the marginal screening approach, we selected the $2p$ interactions with the largest marginal correlation with the response and added them to the design matrix. Then a regular Lasso was performed on the augmented matrix of predictors.

Additionally, in scenarios 3--5, we applied the Lasso with all main effects and only the true interactions. This theoretical Oracle approach provided a gold standard against which to test the performance of Backtracking.

We used the procedures mentioned to yield active sets on which we applied OLS to give a final estimator. To select the tuning parameters of the methods we used cross-validation randomly selection 5 folds but repeating this a total of 5 times to reduce the variance of the cross-validation scores. Thus for each $\lambda$ value we obtained an estimate of the expected prediction error that was an average over the observed prediction errors on 25 (overlapping) validation sets of size $n/5 = 50$. Note that for both Backtracking and the iterated Lasso, this form of cross-validation chose not just a $\lambda$ value but also a path rank. When using Backtracking, the size of the active set was restricted to 50 and the size of $C_k$ to $p + 50 \times 49 
/ 2 = 1225$, so $T$ was at most 50.

In scenarios 1 and 2, the results of the methods were almost indistinguishable except that the screening approach performed far worse in scenario 1 where it tended to select several false interactions which in turn hampered the selection of main effects and resulted in a much larger prediction error. 

The results of scenarios 3--5, where the signal contains interactions, are more interesting and given in Table~\ref{tab:sim_results}. For each scenario, method and SNR level, we report 5 statistics. `$L_2$-sq' is the expected squared distance of the signal $\mb{f}^*$ and our prediction functions $\hat{\mb f}$ based on training data $(\mb Y_\mathrm{train}, \mb X_\mathrm{train})$, evaluated at a random independent test observation $\mb x_{\mathrm{new}}$:
\[
 \E_{\mb x_{\mathrm{new}}, \mb Y_{\mathrm{train}}, \mb X_{\mathrm{train}}}(\mb{f}^* \{\mb x_{\mathrm{new}}) - \hat{\mb f} (\mb x_{\mathrm{new}}; \mb Y_{\mathrm{train}}, \mb X_{\mathrm{train}})\}^2.
\]
`FP Main' and `FP Inter' are the numbers of noise main effects and noise interaction terms respectively, incorrectly included in the final active set. `FN Main' and `FN Inter' are the numbers of true main effects and interaction terms respectively, incorrectly excluded from the final active set.

For all the statistics presented, lower numbers are to be preferred. However, the higher number of false selections incurred by both Backtracking and the Oracle procedure compared to using the main effects only or iterated Lasso fits, is due to the model selection criterion being the expected prediction error. It should not be taken as an indication that the latter procedures are  performing better in these cases.

Backtracking performs best out of the four methods compared here. Note that under all of the settings, iterated Lasso fits incorrectly selects more interaction terms than Backtracking. We see that the more careful way in which Backtracking adds candidate interactions, helps here. Unsurprisingly, fitting the Lasso on just the main effects performs rather poorly in terms of predictive performance. However, it also fails to select important main effects; Backtracking and Iterates have much lower main effect false negatives. The screening approach appears to perform worst here. This is partly because it is not making use of the fact that in all of the examples considered, the main effects involved in interactions are also informative. However, its poor performance is also due the fact that too many false interactions are added to the design matrix after the screening stage. Reducing the number added may help to improve results, but choosing the number of interactions to include via cross-validation, for example, would be computationally costly, unless a Backtracking-type strategy of the sort introduced in this paper were used. We also note that for very large $p$, marginal screening of interactions would be infeasible due to the quadratic scaling in complexity with $p$.

\begin{table}[ht]
\footnotesize
\centering
\begin{tabular}{|c | c | R{0.9cm} R{0.9cm} R{0.9cm} R{0.9cm} R{0.9cm} | R{0.9cm} R{0.9cm} R{0.9cm} R{0.9cm} R{0.9cm}|}
\hline
& & \multicolumn{5}{c|}{$\mathrm{SNR}=2$} & \multicolumn{5}{c|}{$\mathrm{SNR}=3$} \\
\cline{3-12}
Scenario & Statistic & Main & Iterate & Screening & Back-tracking & Oracle & Main & Iterate & Screening & Back-tracking & Oracle \\ 
  \hline
\multirow{5}{*}{3}& $L_2$-sq & 6.95 & 1.40 & 12.87 & 1.21 & 0.82 & 5.67 & 0.27 & 9.24 & 0.27 & 0.18 \\
& FP Main & 3.18 & 2.43 & 0.01 & 2.89 & 3.19 & 1.91 & 0.65 & 0.00 & 0.73 & 0.79 \\
& FN Main & 1.26 & 0.38 & 7.24 & 0.24 & 0.14 & 0.52 & 0.05 & 5.14 & 0.04 & 0.01 \\ 
& FP Inter & 0.00 & 0.93 & 11.05 & 0.45 & 0.00 & 0.00 & 0.27 & 13.57 & 0.12 & 0.00 \\ 
& FN Inter & 3.00 & 0.18 & 2.06 & 0.14 & 0.01 & 3.00 & 0.03 & 1.39 & 0.04 & 0.00 \\ 
\hline
\multirow{5}{*}{4}& $L_2$-sq & 12.05 & 3.25 & 17.68 & 2.72 & 1.68 & 10.44 & 0.63 & 15.19 & 0.41 & 0.31 \\ 
& FP Main & 2.22 & 3.88 & 0.02 & 5.34 & 7.05 & 2.58 & 1.80 & 0.04 & 2.08 & 2.21 \\ 
 & FN Main & 3.12 & 0.90 & 8.13 & 0.61 & 0.26 & 1.77 & 0.11 & 6.94 & 0.04 & 0.00 \\
 & FP Inter & 0.00 & 2.50 & 12.33 & 0.77 & 0.00 & 0.00 & 1.77 & 17.90 & 0.28 & 0.00 \\ 
 & FN Inter & 5.00 & 0.66 & 4.07 & 0.51 & 0.08 & 5.00 & 0.08 & 3.39 & 0.03 & 0.00 \\ 
\hline
 \multirow{5}{*}{5}& $L_2$-sq & 14.12 & 5.08 & 19.96 & 4.52 & 2.14 & 12.84 & 1.56 & 16.99 & 1.17 & 0.44 \\ 
 & FP Main & 3.07 & 4.75 & 0.02 & 5.87 & 8.57 & 3.43 & 3.01 & 0.05 & 3.23 & 3.77 \\ 
 & FN Main & 3.20 & 1.26 & 8.26 & 0.98 & 0.33 & 2.35 & 0.25 & 7.00 & 0.19 & 0.02 \\ 
 & FP Inter & 0.00 & 3.28 & 17.97 & 0.87 & 0.00 & 0.00 & 3.05 & 21.92 & 0.55 & 0.00 \\
 & FN Inter & 6.00 & 1.34 & 5.00 & 1.23 & 0.14 & 6.00 & 0.39 & 4.14 & 0.30 & 0.00 \\ 
   \hline
\end{tabular}
\caption{Simulation results.}
\label{tab:sim_results}
\end{table}

\subsection{Data Analyses} \label{sec:real_data}
In this section, we look at the performance of Backtracking using two base procedures, the Lasso for the linear model and the Lasso for multinomial regression, on a regression and a classification data set. As competing methods, we consider simply using the base procedures (`Main'), iterated Lasso fits (`Iterated'), Lasso following marginal screening for interactions (`Screening'), Random Forests \citep{breiman01random}, hierNet \citep{Bien_etal2013} and MARS \citep{MARS} (implemented using \citet{mda}). Note that we do not view the latter two methods as competitors of Backtracking, as they are designed for use on lower dimensional data sets than Backtracking is capable of handling. However, it is still interesting to see how the methods perform on data of dimension that is perhaps approaching the upper end of what is easily manageable for methods such as hierNet and MARS, but at the lower end of what one might use Backtracking on.

Below we describe the data sets used which are both from the UCI machine learning repository \citep{UCI}.

\subsubsection{Communities and Crime}
This data set available at \url{http://archive.ics.uci.edu/ml/data sets/Communities+and+Crime+Unnormalized} contains crime statistics for the year 1995 obtained from FBI data, and national census data from 1990, for various towns and communities around the USA. We took violent crimes per capita as our response: violent crime being defined as murder, rape, robbery, or assault. The data set contains two different estimates of the populations of the communities: those from the 1990 census and those from the FBI database in 1995. The latter was used to calculate our desired response using the number of cases of violent crimes. However, in several cases, the FBI population data seemed suspect and we discarded all observations where the maximum of the ratios of the two available population estimates differed by more than 1.25. In addition, we  removed all observations that were missing a response and several variables for which the majority of values were missing.
 This 
resulted in a data set with 
$n = 1903$ observations and $p=101$ covariates. The response was scaled to have empirical variance 1.

\subsubsection{ISOLET}
This data set consists of $p=617$ features based on the speech waveforms generated from utterances of each letter of the English alphabet. The task is to learn a classifier which can determine the letter spoken based on these features. The data set is available from \url{http://archive.ics.uci.edu/ml/datasets/ISOLET}; see \citet{ISOLET} for more background on the data. We consider classification on the notoriously challenging E-set consisting of the letters `B', `C', `D', `E', `G', `P', `T', `V' and `Z' (pronounced `zee'). As there were 150 subjects and each spoke each letter twice, we have $n = 2700$ observations spread equally among 9 classes.
The dimension of this data is such that MARS and hierNet could not be applied.

\subsection{Methods and Results}
For the Communities and crime data set, we used the Lasso for the linear model as the base regression procedure for Backtracking and Iterates. Since the per capita violent crime response was always non-negative, the positive part of the fitted values was taken. For Main, Backtracking, Iterates, Screening and hierNet, we employed 5-fold cross-validation with squared error loss to select tuning parameters. For MARS we used the default settings for pruning the final fits using generalised cross-validation. With Random Forests, we used the default settings on both data sets.
For the classification example, penalised multinomial regression was used (see Section~\ref{sec:multinomial}) as the base procedure for Backtracking and Iterates, and the deviance was used as the loss function for 5-fold cross-validation. In all of the methods except Random Forests, we only included first-order interactions. When using Backtracking, we also restricted the size of $C_k$ to $p + 50 \times 49 / 2 = p + 1225$.

To evaluate the procedures, we randomly selected 2/3 for training and the remaining 1/3 was used for testing. This was repeated 200 times for each of the data sets. Note that we have specifically chosen data sets with $n$ large as well as $p$ large. This is to ensure that comparisons between the performances of the methods can be made with more accuracy. For the regression example, out-of-sample squared prediction error was used as a measure of error; for the classification example, we used out-of-sample misclassification error with 0--1 loss. The results are given in Table~\ref{tab:real_data}.

Random Forests has the lowest prediction error on the regression data set, with Backtracking not far behind, whilst Backtracking wins in the classification task, and in fact achieves strictly lower misclassification error than all the other methods on 90\% of all test samples. Note that a direct comparison with Random Forests is perhaps unfair, as the latter is a black-box procedure whereas Backtracking is aiming for a more interpretable model.

MARS performs very poorly indeed on the regression data set. The enormous prediction error is caused by the fact that whenever observations corresponding to either New York or Los Angeles were in the test set, MARS predicted their responses to be far larger than they were. However, even with these observations removed, the instability of MARS meant that it was unable to give much better predictions than an intercept-only model.

HierNet performs well on this data set, though it is worth noting that we had to scale the interactions to have the same $\ell_2$-norm as the main effects to get such good results (the default scaling produced error rates worse than that of an intercept-only model). Backtracking does better here. One reason for this is that the because the main effects are reasonably strong in this case, a low amount of penalisation works well. However, because with hierNet, the penalty on the interactions is coupled with the penalty on the main effects, the final model tended to include close to two hundred interaction terms. The Screening approach similarly suffers from including too many interactions and performs only a little better than a main effects only fit.

The way that Backtracking creates several solution paths with varying numbers of interaction terms means that it is possible to fit main effects and a few interactions using a low penalty without this low penalisation opening the door to many other interaction terms. The iterated Lasso approach also has this advantage, but as the number of interactions are increased in discrete stages, it can miss a candidate set with the right number of interactions that may be picked up by the more continuous model building process used by Backtracking. This occurs in a rather extreme way with the ISOLET data set where, since in the first stage of the iterated Lasso, cross-validation selected far too many variables ($>250$), the second and subsequent steps could not be performed. This is why the results are identical to using the main effects alone.

\begin{table}
 \centering
  \begin{tabular}{|c| L{4.5cm} | L{4.5cm}|}
   \hline
& \multicolumn{2}{c|}{Error} \\
\cline{2-3}
Method & \multicolumn{1}{c|}{Communities and crime} & \multicolumn{1}{c|}{ISOLET} \\
\hline
Main & {\color{white}00000}0.414 \;$(6.5\times 10^{-3})$ & {\color{white}0000}0.0641 \;$(4.7\times 10^{-4})$ \\
Iterate & {\color{white}00000}0.384 \;$(5.9\times 10^{-3}$ & {\color{white}0000}0.0641 \;$(4.7\times 10^{-4})$ \\
Screening & {\color{white}00000}0.390 \;$(7.8\times 10^{-3})$ & \multicolumn{1}{c|}{-} \\
Backtracking & {\color{white}00000}0.365 \;$(3.7\times 10^{-3})$ & {\color{white}0000}0.0563 \;$(4.5\times 10^{-4})$ \\
Random Forest & {\color{white}00000}0.356 \;$(2.4\times 10^{-3})$ & {\color{white}0000}0.0837 \;$(6.0\times 10^{-4})$ \\
hierNet & {\color{white}00000}0.373 \;$(4.7\times 10^{-3})$ &   \multicolumn{1}{c|}{-} \\
MARS & {\color{white}00}5580.586 \;$(3.1 \times 10^3)$ & \multicolumn{1}{c|}{-} \\
\hline
  \end{tabular}
\
\caption{Real data analyses results. Average error rates over 200 training--testing splits are given, with standard deviations of the results divided by $\sqrt{200}$ in parentheses.}
\label{tab:real_data}
\end{table}

\section{Theoretical Properties} \label{sec:theory}
Our goal in this section is to understand under what circumstances Backtracking with the Lasso can arrive at a set of candidates, $C^*$, that contains all of the true interactions, and only a few false interactions.  On the event on which this occurs, we can then apply
many of the existing results on the Lasso, to show that the solution path $\hat{\mbb\beta}(\lambda, C^*)$ has certain properties. As an example, in Section~\ref{sec:fixed} we give sufficient conditions for the existence of a $\lambda^*$ such that $\{v: \hat{\beta}_v (\lambda^*, C^*) \neq 0\}$ equals the true set of variables.

We work with the normal linear model with interactions,
\begin{equation} \label{eq:mod}
 \mb Y = \mu^* \textbf{1} + \mb X_{S^*} \mbb\beta^*_{S^*} + \mbb\varepsilon,
\end{equation}
where $\varepsilon_i \iid N(0, \sigma^2)$, and to ensure identifiability, $\mb X_{S^*}$ has full column rank. We will assume that $S^* = S^* _1 \cup S^* _2$, where $S^* _1$ and $S^* _2$ are main effects and two-way interactions respectively. Let the interacting main effects be $I^*$; formally, $I^*$ is the smallest set of main effects such that $\mathcal{I}(I^*) \supseteq S^*_2$.
Assume $I^* \subseteq S^*_1$ so interactions only involve important main effects.
Let $s_l=|S^*_l|$, $l=1,2$ and set $s=s_1+s_2$. Define $C^*=C_1 \cup \mathcal{I}(S_1^*)$. Note that $C^*$ contains $S^*$ but not additional interactions from any variables from $C_1 \setminus S_1^*$.  

Although the Backtracking algorithm was presented for a base path algorithm that computed solutions at only discrete values, for the following results, we need to imagine an idealised algorithm which computes the entire path of solutions.
In addition, we will assume that we only allow first-order interactions in the Backtracking algorithm, and that $T \geq s_1$.

We first consider the special case where the design matrix is derived from a random matrix with i.i.d.\ multivariate normal rows, before describing a result for fixed design. 

\subsection{Random Normal Design}
Let the random matrix $\mb Z$ have independent rows distributed as $N_p(\mb 0, \mbb\Sigma)$. Suppose that $\mb X_{C_1}$, the matrix of main effects, is formed by scaling and centring $\mb Z$.
We consider an asymptotic regime where $\mb X$, $\mb f^*$,  $S^*$, $\sigma^2$ and $p$ can all change as $n \to \infty$, though we will suppress their dependence on $n$ in the notation.
Furthermore, for sets of indices $S, M \subseteq \{1,\ldots,p\}$, let $\mbb\Sigma_{S,M} \in \R^{|S| \times |M|}$ denote the submatrix of $\mbb\Sigma$ formed from those rows and columns of $\mbb\Sigma$ indexed by $S$ and $M$ respectively.
For any positive semi-definite matrix $\mb A$, we will let $c_{\text{min}}(\mb A)$ and $c_{\text{max}}(\mb A)$ denote its minimal and maximal eigenvalues respectively. For sequences $a_n$, $b_n$, by $a_n \succ b_n$ we mean $b_n=o(a_n)$.
We make the following assumptions.
\begin{itemize}
\item[A1.] $c_{\text{min}}(\mbb\Sigma_{S_1^*,S_1^*}) \geq c_*>0$.
\item[A2.] $\sup_{\mbb\tau \in \R^{s_1}:\|\mbb\tau\|_\infty\leq 1} \|\mbb\Sigma_{N,S^*_1}\mbb\Sigma_{S^*_1,S^*_1}^{-1}\mbb\tau\|_\infty \leq \delta < 1$.
\item[A3.] $s_1^4 \log(p)/n \to 0$ and $s_1^8 \log(s_1)^2/n \to 0$.
\item[A4.] \[
\min_{j \in I^*} |\beta^*_j| \succ \frac{s_1(\sigma\sqrt{\log p} + \sqrt{s_1  +\log p})}{\sqrt{n}} + \frac{\sqrt{s_1^3\log(s_1)}}{n^{1/3}}.
\]
\item[A5.] $\|\mbb\beta^*_{S_2^*}\|_2$ is bounded as $n \to \infty$ and $c_{\text{max}}(\mbb\Sigma_{S_1^*,S_1^*}) \leq c^*<\infty$.
\end{itemize}
A1 is a standard assumption in high-dimensional regression and is, for example, implied by the compatibility constant of \citet{BuhlmannGeer2009} being bounded away from zero. A2 is closely related to irrepresentable conditions (see \citet{meinshausen04consistent}, \citet{zhao05model}, \citet{zou05adaptive}, \citet{BuhlmannGeer2009}, \citet{wainwright06sth}), which are used for proving variable selection consistency of the Lasso. Note that although here the signal may contain interactions, our irrepresentable-type condition only involves main effects.

A3 places restrictions on the rates at which $s_1$ and $p$ can increase with $n$. The first condition involving $\log(p)$ is somewhat natural as $s_1^2\log(p)/n \to 0$ would typically be required in order to show $\ell_1$ estimation consistency of $\mbb\beta$ where only $s_1$ main effects are present; here our effective number of variables is $s_1 \leq s \leq s_1^2$. The second condition restricts the size of $s_1$ more stringently but is nevertheless weaker than equivalent conditions in \citet{Hao2014}.

A4 is a minimal signal strength condition. The term involving $\sigma$ is the usual bound on the signal strength required in results on variable selection consistency with the Lasso when there are $s_1^2$ non-zero variables. Due to the presence of interactions, the terms not involving $\sigma$ place additional restrictions on the sizes of non-zero components of $\mbb\beta^*$ even when $\sigma=0$.
A5 ensures that the model is not too heavily misspecified in the initial stages of the algorithm, where we are regressing on only main effects.

The following theorem states that given the assumptions above, with probability tending to 1 we are guaranteed a candidate set will be produced by our algorithm which contains all true interactions and no interactions involving a noise variable.
\begin{thm} \label{thm:random}
Assuming A1--A5, the probability that there exists a $k^*$ such that $C^* \supseteq C_{k^*} \supseteq S^*$ tends to 1 as $n \to \infty$.
\end{thm}

\subsection{Fixed Design} \label{sec:fixed}
The result for a random normal design above is based on a corresponding result for fixed design which we present here.
In order for Backtracking not to add any interactions involving noise variables, to begin with, one pair of interacting signal variables must enter the solution path before any noise variables. 
Other interacting signal variables need only become active after the interaction between this first pair has become active. Thus we need that there is some ordering of the interacting variables where each variable only requires interactions between those variables earlier in the order to be present before it can become active. Variables early on in the order must have the ability to be selected when there is serious model misspecification as few interaction terms will be available for selection. Variables later in the order only need to have the ability to be selected when the model is approximately correct.

Note that a signal variable having a coefficient large in absolute value does not necessarily ensure that it becomes active before any noise variable. Indeed, in our example in Section~\ref{sec:motivation}, variable 5 did not enter the solution path at all when only main effects were present, but had the largest coefficient.
Write $\mb{f}^*$ for $\mb X_{S^*} \mbb\beta_{S^*}$, and for a set $S$ such that $\mb X_S$ has full column rank, define
\[
\mbb \beta ^S := (\mb X_S ^T \mb X_S)^{-1} \mb X_S ^T \mb{f}^*.
\]
Intuitively what should matter are the sizes of the appropriate coefficients of $\mbb \beta ^S$ for  suitable choices of $S$. In the next section, we give a sufficient condition based on $\mbb\beta^S$ for a variable $v \in S$ to enter the solution path before any variable outside $S$.

\subsubsection{The Entry Condition} \label{sec:entry}
Let $\mb P^S = \mb X_S (\mb X_S ^T \mb X_S)^{-1} \mb X_S ^T$ denote orthogonal projection on to the space spanned by the columns of $\mb X_S$.
Further, for any two candidate sets $S, M$ that are sets of subsets of $\{1,\ldots,p\}$,
define
\begin{gather*}
  \hat{\mbb\Sigma}_{S, M} = \tfrac{1}{n} \mb X_S ^T \mb X_M.
\end{gather*}
Now given a set of candidates, $C$, let $v \in S \subset C$ and write $M = C \setminus S$. For $\eta > 0$, we shall say that the $\mathrm{Ent}(v, S, C; \eta)$ condition holds if, $\mb X_S$ has full column rank, and the following holds,
\begin{gather} 
\sup_{\mbb\tau_S \in \R^{|S|} : \norms{\mbb\tau_S}_\infty \leq 1} \norms{\hat{\mbb\Sigma}_{M, S} \hat{\mbb\Sigma}_{S, S} ^{-1} \mbb\tau_S}_\infty < 1, \label{eq:irrep} \\
 |\beta^S _v| > \max_{u \in M} \left\{ \frac{\tfrac{1}{n}\abs{\mb X_u ^T (\mb I - \mb P^S) \mb{f}^*} + 2\eta}{1 - \norms{\hat{\mbb\Sigma}_{S, S} ^{-1} \hat{\mbb\Sigma}_{S, \{u\}}}_1} + \eta \right\} \norms{(\hat{\mbb\Sigma}_{S, S} ^{-1})_v}_1 .\label{eq:ent}
\end{gather}
In Lemma~\ref{lem:entry} given in the appendix, we show that this condition is sufficient for variable $v$ to enter the active set before any variable in $M$, when the set of candidates is $C$ and $\norms{\mb X_C ^T \mbb\varepsilon}_\infty \leq \eta$. In addition, we show that $v$ will remain in the active set at least until some variable from $M$ enters the active set.

The second part of the entry condition \eqref{eq:ent} asserts that coefficient $v$ of the regression of $\mb{f}^*$ on $\mb X_S$ must exceed a certain quantity that we now examine in more detail. The $\tfrac{1}{n} \mb X_u ^T (\mb I - \mb P^S) \mb{f}^*$ term is the sample covariance between $\mb X_u$, which is one of the columns of $\mb X_M$, and the residual from regressing $\mb{f}^*$ on $\mb X_S$. Note that the more of $S^*$ that $S$ contains, the closer this will be to 0.

To understand the $\norms{(\hat{\mbb\Sigma}_{S, S} ^{-1})_v}_1$ term, without loss of generality take $v$ as $\{1\}$ and write $\mb b = \hat{\mbb\Sigma}_{S \setminus \{v\}, \{v\}}$ and $\mb D = \hat{\mbb\Sigma}_{S \setminus \{v\}, S \setminus \{v\}}$. For any square matrix $\hat{\mbb\Sigma}$, let $c_{\mathrm{min}} (\hat{\mbb\Sigma})$ denote its minimal eigenvalue.
Using the formula for the inverse of a block matrix and writing $s$ for $|S|$, we have
\begin{align*}
 \norms{(\hat{\mbb\Sigma}_{S, S} ^{-1})_v}_1 & = \norm{\begin{pmatrix} 1 + \mb b^T(\mb D - \mb b \mb b^T)^{-1} \mb b \\ -(\mb D - \mb b \mb b^T)^{-1} \mb b \end{pmatrix}}_1 \\
& \leq 1 + \frac{\norms{\mb b}^2 _2 + \sqrt{s-1}\norms{\mb b}_2}{c_\mathrm{min} (\hat{\mbb\Sigma}_{S, S})}.
\end{align*}
In the final line we have used the Cauchy--Schwarz inequality and the fact that if $\mb w^*$ is a unit eigenvector of $\mb D - \mb b \mb b^T$ with minimal eigenvalue, then
\[
 c_\mathrm{min} (\mb D - \mb b \mb b^T) = \norm{ \hat{\mbb\Sigma}_{S, S} \begin{pmatrix}
						-\mb b^T \mb w^* \\ \mb w^* 
						\end{pmatrix}}_2  \geq c_\mathrm{min} (\hat{\mbb\Sigma}_{S, S})\sqrt{1 + |\mb b^T \mb w^*|^2} \geq c_\mathrm{min}(\hat{\mbb\Sigma}_{S, S}).
\]
Thus when variable $v$ is not too correlated with the other variables in $S$, and so $\norms{\mb b}_2$ is small, $\norms{(\hat{\mbb\Sigma}_{S, S} ^{-1})_v}_1$ will not be too large. Even when this is not the case, we still have the bound
\[
 \norms{(\hat{\mbb\Sigma}_{S, S} ^{-1})_v}_1 \leq \frac{\sqrt{|S|}}{c_\mathrm{min} (\hat{\mbb\Sigma}_{S, S})}.
\]

Turning now to the denominator, $\norms{\hat{\mbb\Sigma}_{S, S} ^{-1} \hat{\mbb\Sigma}_{S, \{u\}}}_1$ is the $\ell_1$-norm of the coefficient of regression of $\mb X_u$ on $\mb X_S$, and the maximum of this quantity over $u \in M$ gives the left-hand side of \eqref{eq:irrep}. Thus when $u$ is highly correlated with many of the variables in $S$, $\norms{\hat{\mbb\Sigma}_{S, S} ^{-1} \hat{\mbb\Sigma}_{S, \{u\}}}_1$ will be large. On the other hand, in this case one would expect $\norms{(\mb I - \mb P^S) \mb X_u}_2$ to be small, and so to some extent the numerator and denominator compensate for each other.

\subsubsection{Statement of Results}
Without loss of generality assume $I^* = \{1, \ldots, |I^*|\}$. Also let $\mathcal{J}=\{\mathcal{I}(A):A\subseteq S_1^*\}$. Our formal assumption corresponding to the discussion at the beginning of Section~\ref{sec:theory} is the following.
\begin{quote}
\textbf{The entry order condition.} There is some ordering of the variables in $I^*$, which without loss of generality we take to simply be $1, \ldots, |I^*|$, such that for each $j \in I^*$, we have,
\begin{gather*}
 \text{For all } A \in \mathcal{J} \text{ with }\mathcal{I}(1, \ldots, j-1) \subseteq A \subseteq \mathcal{I}(S_1 ^*) \nonumber \\
\mathrm{Ent}(j, S^*_1 \cup B, C_1 \cup A; \eta) \,\, \text{ holds for some } A \cap S_2^* \subseteq B \subseteq A.
\end{gather*}
Here
\begin{gather*}
 \eta = \eta(t; n, p, s_1, \sigma) = \sigma \sqrt{\frac{t^2 + 2 \log(p + s_1^2)}{n}}.
\end{gather*}
\end{quote}
First we discuss the implications for variable $1$. The condition ensures that whenever the candidate set is enlarged from $C_1$ to also include any set of interactions built from $S_1^*$, variable $1$ enters the active set before any variable outside $\mathcal{I}(S_1^*)$, and moreover, it remains in the active set at least until a variable outside $\mathcal{I}(S_1^*)$ enters.

For $j >2$, we see that the enlarged candidate sets for which we require the entry conditions to hold, are fewer in number. Variable $|I^*|$ only requires the entry condition to hold for candidate sets that at least include $\mathcal{I}(1,\ldots, |I^*|-1)$ and thus include almost all of $S^*$. What this means is that we require some `strong' interacting variables, for which when $\mb{f}^*$ is regressed onto a variety of sets of variables containing them (some of which contain only a few of the true interaction variables), always have large coefficients. Given the existence of such 
strong variables, other interacting variables need only have large coefficients when $\mb{f}^*$ is regressed onto sets containing them that also include many true interaction terms. Note that the equivalent result for the success of the strategy that simply adds interactions between selected main effects would essentially require all main effect involved in interactions to satisfy the conditions imposed on the variables $1$ and $2$ here.
Going back to the example in Section~\ref{sec:motivation}, variable 5 has $|\beta_{5} ^S| \approx 0$ for all $S \subseteq \{1, \ldots, 6\}$, but $|\beta_{5} ^S| > 0$ once $\{1,2\} \in S$ or $\{3,4\} \in S$. 



\begin{thm} \label{thm:main}
Assume the entry order condition holds. With probability at least $1 - \exp(-t^2/2)$, there exists a $k^*$ such that $C^* \supseteq C_{k^*} \supseteq S^*$.
\end{thm}
The following corollary establishes variable selection consistency under some additional conditions. 
\begin{cor} \label{cor:var_sel}
Assume the entry order condition holds. Writing $N = C^* \setminus S^*$, further assume 
\[\norms{\hat{\mbb\Sigma}_{N, S^*} \hat{\mbb\Sigma}_{S^*, S^*} ^{-1} \sgn(\mbb\beta^* _{S^*})}_\infty < 1 ;\]
and that for all $v \in S^*$,
\begin{equation*}
 | \beta^* _v| > \frac{\eta \abs{\sgn(\mbb\beta^* _{S^*})^T (\hat{\mbb\Sigma}_{S^*, S^*} ^{-1} )_v}}{1 - \norms{\hat{\mbb\Sigma}_{N, S^*} \hat{\mbb\Sigma}_{S^*, S^*} ^{-1} \sgn(\mbb\beta^* _{S^*})}_\infty} + \xi,
\end{equation*}
where
\[
 \xi = \xi(t; n, s, \sigma, c_{\mathrm{min}} (\hat{\mbb\Sigma}_{S^*, S^*})) = \sigma \sqrt{\frac{t^2 + 2\log (s)}{n c_{\mathrm{min}}(\hat{\mbb\Sigma}_{S^*, S^*})}}.
\]
Then with probability at least $1 - 3 \exp(-t^2 / 2)$, there exist $k^*$ and $\lambda^*$ such that
\[
 \mathcal{A}(\tilde{\mbb\beta}(\lambda^*, C_{k^*})) = S^*.
\]
\end{cor}
Note that if we were to simply apply the Lasso to the set of candidates $C^\mathrm{all} := C_1 \cup \mathcal{I}(C_1)$ (i.e. all possible main effects and their first-order interactions), we would require an irrepresentable condition of the form
\[
 \norms{\hat{\mbb\Sigma}_{N^\mathrm{all}, S^*} \hat{\mbb\Sigma}_{S^*, S^*} ^{-1} \sgn(\mbb\beta^* _{S^*})}_\infty < 1,
\]
where $N^\mathrm{all} = C^\mathrm{all} \setminus S^*$. Thus we would need $O(p^2)$ inequalities to hold, rather than
 our $O(p)$. Of course, we had to introduce many additional assumptions to reach this stage and no set of assumptions is uniformly stronger or weaker than the other. However, our proposed method is computationally feasible.
 
\section{Discussion} \label{sec:discussion}
 While several methods now exist for fitting interactions in moderate-dimensional situations where $p$ is in the order of hundreds, the problem of fitting interactions when the data is of truly high dimension has received less attention.
 
Typically, the search for interactions must be restricted by first fitting a model using only main effects, and then including interactions between those selected main effects, as well as the original main effects, as candidates in a final fit. This approach has the drawbacks that important main effects may not be selected in the initial stage as they require certain interactions to be present in order for them to be useful for prediction. In addition, the initial model may contain too many main effects when, without the relevant interactions, the model selection procedure cannot find a good sparse approximation to the true model.

The Backtracking method proposed in this paper allows interactions to be added in a more natural gradual fashion, so there is a better chance of having a model which contains the right interactions. The method is computationally efficient, and our numerical results demonstrate its effectiveness for both variable selection and prediction.

From a theoretical point of view we have shown that when used with the Lasso, rather than requiring all main effects involved in interactions to be highly correlated with the signal, Backtracking only needs there to exist some ordering of these variables where those early on in the order are important for predicting the response by themselves. Variables later in the order only need to be helpful for predicting the response when interactions between variables early on in the order are present.

Though in this paper, we have largely focussed on Backtracking used with the Lasso, the method is very general and can be used with many procedures that involve sparsity-inducing penalty functions. These methods tend to be some of the most useful for dealing with high-dimensional data, as they can produce stable, interpretable models. Combined with Backtracking, the methods become much more flexible, and it would be very interesting to explore to what extent using non-linear base procedures could yield interpretable models with predictive power comparable to black-box procedures such as Random Forests \citep{breiman01random}. In addition, we believe integrating Backtracking with some of the penalty-based methods for fitting interactions to moderate-dimensional data, will prove to be a fruitful direction for future research.
 
\section*{Acknowledgements}
I am very grateful to Richard Samworth, for many helpful comments and suggestions.

\appendix
%

\section{Construction of $\mb X$ in Section~\ref{sec:motivation}}
First, consider $(Z_{i1}, Z_{i2}, Z_{i3})$ generated from a mean zero multivariate normal distribution with $\Var(Z_{ij}) = 1$, $j = 1, 2, 3$, $\Cov(Z_{i1}, Z_{i2}) = 0$ and $\Cov(Z_{i1}, Z_{i3}) = \Cov(Z_{i2}, Z_{i3}) = 1/2$. Independently generate $R_{i1}$ and $R_{i2}$ each of which takes only the values $\{-1, 1\}$, each with probability $1/2$. We form the $i$th row of the design matrix as follows:
\begin{align*}
 X_{i1} =&  R_{i1} \,\sgn(Z_{i1}) |Z_{i1}|^{1/4}, \\
X_{i2} =&  R_{i1}  |Z_{i1}|^{3/4}, \\
X_{i3} =&  R_{i2} \,\sgn(Z_{i2}) |Z_{i2}|^{1/4}, \\
X_{i4} =&  R_{i2} |Z_{i2}|^{3/4}, \\
X_{i5} =& Z_{i3}.
\end{align*}
The remaining $X_{ij}$, $j=6,\ldots,p$ are independently generated from a standard normal distribution.
Note that the random signs $R_{i1}$ and $R_{i2}$ ensure that $X_{i5}$ is uncorrelated with each of $X_{i1}, \ldots, X_{i4}$. Furthermore, the fact that $X_{i1} X_{i2} = Z_{i1}$ and $X_{i3} X_{i4} = Z_{i2}$, means that when $\beta_5 = -\tfrac{1}{2}(\beta_7 + \beta_8)$, $X_{i5}$ is uncorrelated with the response.

\section{Proofs of Theorem~\ref{thm:main} and Corollory~\ref{cor:var_sel}}
In this subsection we use many ideas from Section B of \citet{wainwright06sth} and Section 6 of \citet{BuhlmannGeer2009}.
\begin{lem} \label{lem:entry}
 Let $S \subseteq C$ be such that $X_S$ has full column rank and let $M = C \setminus S$. On the event
\begin{equation*}
 \Omega_{C, \eta} := \{ \tfrac{1}{n} \norms{\mb X_C ^T \mbb\varepsilon}_\infty \leq \eta \},
\end{equation*}
the following hold:
\begin{enumerate}[(i)]
 \item If
\begin{equation} \label{eq:lambda_cond}
 \lambda > \max_{u \in M} \left\{ \frac{\tfrac{1}{n}|\mb X_u ^T (\mb I - \mb P^S) \mb{f}^*| + 2\eta}{1 - \norms{\hat{\mbb\Sigma}_{S, S} ^{-1} \hat{\mbb\Sigma}_{S, \{u\}}}_1} \right\},
\end{equation}
then the Lasso solution is unique and $\hat{\mbb\beta}_M (\lambda, C) = \mb 0$.
\item If $\lambda$ is such that for some Lasso solution $\hat{\mbb\beta}_M (\lambda, C) = \mb 0$, and for $v \in S$,
\begin{equation*}
 |\beta^S _v| > \norms{(\hat{\mbb\Sigma}_{S, S} ^{-1})_v}_1 (\lambda + \eta),
\end{equation*}
then for all Lasso solutions, $\hat{\beta}_v (\lambda, C) \neq 0$.
\item Let 
\begin{equation*}
 \lambda^\mathrm{ent} = \sup \{\lambda : \lambda \geq 0 \text{ and for some Lasso solution }\hat{\mbb\beta}_M (\lambda, C) \neq \mb 0\},
\end{equation*}
where we take $\sup \emptyset = 0$. If for $v \in S$,
\begin{equation*}
 |\beta^S _v| > \max_{u \in M} \left\{ \frac{\tfrac{1}{n}|\mb X_u^T (\mb I - \mb P^S) \mb{f}^*| + 2\eta}{1 - \norms{\hat{\mbb\Sigma}_{S, S} ^{-1} \hat{\mbb\Sigma}_{S, \{u\}}}_1} + \eta \right\} \norms{(\hat{\mbb\Sigma}_{S, S} ^{-1})_v}_1,
\end{equation*}
there exists a $\lambda > \lambda^\mathrm{ent}$ such that the solution $\hat{\mbb \beta}(\lambda, C)$ is unique, and for all $\lambda' \in (\lambda^\mathrm{ent}, \lambda]$ and all Lasso solutions $\hat{\mbb \beta}(\lambda', C)$, we have $\hat{\beta}_v (\lambda', C) \neq 0.$
\end{enumerate}
\end{lem}
\begin{proof}
We begin by proving (i).
Suppressing the dependence of $\hat{\mbb \beta}$ on $\lambda$ and $C$, we can write the KKT conditions (\eqref{eq:KKT_non_zero}, \eqref{eq:KKT_zero}) as
\[
 \frac{1}{n} \mb X_C ^T (\mb Y - \mb X_C \hat{\mbb \beta}) = \lambda \hat{\mbb\tau},
\]
where $\hat{\mbb\tau}$ is an element of the subdifferential $\partial \norms{\hat{\mbb \beta}}_1$ and thus satisfies
\begin{gather}
 \norms{\hat{\mbb\tau}}_\infty \leq 1, \label{eq:infty_norm} \\
\hat{\beta}_v \neq 0 \Rightarrow \hat{\tau}_v = \sgn(\hat{\beta}_v). \label{eq:l1_norm_equiv}
\end{gather}
By decomposing $\mb Y$ as $\mb P^S \mb{f}^* + (\mb I - \mb P^S)\mb{f}^* + \mb \varepsilon$, $\mb X_C$ as $(\mb X_S \, \mb X_M)$, and noting that $\mb X_S ^T (\mb I - \mb P^S) = \mb 0$, we can rewrite the KKT conditions in the following way:
\begin{gather}
 \tfrac{1}{n} \mb X_S ^T(\mb P^S \mb{f}^*  - \mb X_S \hat{\mbb\beta}_S) + \tfrac{1}{n} \mb X_S ^T \mbb\varepsilon - \hat{\mbb\Sigma}_{S, M} \hat{\mbb\beta}_{J^*} = \lambda \hat{\mbb\tau}_{S}, \label{eq:KKT_1} \\
 \tfrac{1}{n} \mb X_M ^T(\mb P^S \mb{f}^*  - \mb X_S \hat{\mbb\beta}_S) + \tfrac{1}{n} \mb  X_M ^T \{(\mb I - \mb P^S) \mb{f}^* + \mbb\varepsilon\} - \hat{\mbb\Sigma}_{M, M} \hat{\mbb\beta}_M = \lambda \hat{\mbb\tau}_M. \label{eq:KKT_2}
\end{gather}
Now let $\breve{\mbb\beta}_S$ be a solution to the restricted Lasso problem,
\[
 (\hat{\mu}, \breve{\mbb\beta}_S ) = \argmin{\mu, \mbb\beta_S} \left\{ \tfrac{1}{2n} \norms{\mb Y - \mu \mathbf{1} - \mb X_S \mbb\beta_S}^2 + \lambda \norms{\mbb\beta_S}_1 \right\}.
\]
The KKT conditions give that $\breve{\mbb\beta}_S$ satisfies
\begin{equation} \label{eq:red_KKT}
  \frac{1}{n} \mb X_S ^T (\mb Y - \mb X_S \breve{\mbb\beta}_S) = \lambda \breve{\mbb\tau}_S,
\end{equation}
where $\breve{\mbb\tau}_S \in \partial \norms{\breve{\mbb\beta}_S}_1$. We now claim that
\begin{gather}
 (\hat{\mbb\beta}_S, \hat{\mbb\beta}_M) = (\breve{\mbb\beta}_S, \mb 0) \label{eq:trial1} \\
 (\hat{\mbb\tau}_S, \hat{\mbb\tau}_M) = \left( \breve{\mbb\tau}_S, \, \hat{\mbb\Sigma}_{M, S} \hat{\mbb\Sigma}_{S, S} ^{-1} (\breve{\mbb\tau}_S - \tfrac{1}{n} \lambda^{-1} \mb X_S ^T \mbb\varepsilon) + \tfrac{1}{n} \lambda^{-1} \mb X_M ^T \{(\mb I - \mb P^S) \mb{f}^* + \mbb\varepsilon\}\right) \label{eq:trial2}
\end{gather}
is the unique solution to \eqref{eq:KKT_1}, \eqref{eq:KKT_2}, \eqref{eq:infty_norm} and \eqref{eq:l1_norm_equiv}. Indeed, as $\breve{\mbb\beta}_S$ solves the reduced Lasso problem, we must have that \eqref{eq:KKT_1} and \eqref{eq:l1_norm_equiv} are satisfied. Multiplying \eqref{eq:KKT_1} by $\mb X_S \hat{\mbb\Sigma}_{S, S} ^{-1}$, setting $\hat{\mbb\beta}_M = \mb 0$ and rearranging gives us that
\begin{equation} \label{eq:beta_S}
 \mb P^S \mb{f}^* - \mb X_S \hat{\mbb\beta}_S = \mb X_S \hat{\mbb\Sigma}_{S, S} ^{-1} (\lambda \hat{\mbb\tau}_S - \tfrac{1}{n} \mb X_S ^T \mbb\varepsilon),
\end{equation}
and substituting this into \eqref{eq:KKT_2} shows that our choice of $\hat{\mbb\tau}_M$ satisfies \eqref{eq:KKT_2}. It remains to check that we have $\norms{\hat{\mbb\tau}_M}_\infty \leq 1$. In fact, we shall show that $\norms{\hat{\mbb\tau}_M}_\infty < 1$. Since we are on $\Omega_{C,\eta}$ and $\norms{\breve{\mbb\tau}_S}_\infty \leq 1$, for $u \in M$ we have
\begin{align*}
 \lambda |\hat{\mbb\tau}_u| & \leq \norms{\hat{\mbb\Sigma}_{S, S} ^{-1} \hat{\mbb\Sigma}_{S, \{u\}}}_1 \left( \lambda \norms{\breve{\mbb\tau}_S}_\infty + \norms{\tfrac{1}{n}\mb X_S ^T \mbb\varepsilon}_\infty \right) + \tfrac{1}{n}\abs{\mb X_u^T (\mb I - \mb P^S) \mb{f}^*} + \tfrac{1}{n}\abs{\mb X_u^T \mbb\varepsilon} \\
 & < \lambda \norms{\hat{\mbb\Sigma}_{S, S} ^{-1} \hat{\mbb\Sigma}_{S, \{u\}}}_1 + \tfrac{1}{n}\abs{\mb X_u^T (\mb I - \mb P^S) \mb{f}^*} + 2\eta \\
& < \lambda,
\end{align*}
where the final inequality follows from \eqref{eq:lambda_cond}. We have shown that there exists a solution, $\hat{\mbb \beta}$, to the Lasso optimisation problem with $\hat{\mbb\beta}_M = 0$. The uniqueness of this solution follows from noting that $\norms{\hat{\mbb\tau}_M}_\infty < 1$, $\mb X_S$ has full column rank and appealing to Lemma 1 of \citet{wainwright06sth}.

For (ii), note that from \eqref{eq:KKT_1}, provided $\hat{\mbb\beta}_M = 0$, we have that
\begin{equation*}
 \hat{\mbb\beta}_S = \mbb\beta^S - \hat{\mbb\Sigma}_{S, S} ^{-1} (\lambda \hat{\mbb\tau}_S - \tfrac{1}{n} \mb X_S^T \mbb \varepsilon ).
\end{equation*}
But by assumption
\begin{align*}
 |\beta^S _v| >  \norms{(\hat{\mbb\Sigma}_{S, S} ^{-1})_v}_1 (\lambda + \eta) \geq \abs{{(\hat{\mbb\Sigma}_{S, S} ^{-1})}_v ^T (\lambda \hat{\mbb\tau}_S - \tfrac{1}{n}\mb X_S ^T \mbb\varepsilon)},
\end{align*}
whence $\hat{\beta}_v \neq 0$.

(iii) follows easily from (i) and (ii).
\end{proof}

\vspace{\abovedisplayshortskip}
\paragraph{Proof of Theorem~\ref{thm:main}.}
In all that follows, we work on the event $\Omega_{C^*, \eta}$ defined in Lemma~\ref{lem:entry}. Using standard bounds for the tails of Gaussian random variables and the union bound, it is easy to show that $\pr (\Omega_1 \cap \Omega_{C^*, \eta}) \geq 1 - \exp(-t^2 /2)$. Let $N=\{1,\ldots,p\}\setminus S_1^*$.

Let $\tilde{T}$ be the number of steps taken by the algorithm: this would typically be $T$, but may be smaller if a perfect fit is reached or if $p < T$ for example.
Let $C_{k}$ be the largest member of $\{C_1, \ldots, C_{\tilde{T}}\}$ satisfying $C_{k} \subseteq C^*$. Such a $C_{k}$ exists since $C_1 \subseteq C^*$.

Now suppose for a contradiction that $C_k \nsupseteq S^*$. Let $j$ be such that
\[
 \mathcal{I}(1, \ldots, j-1) \subseteq C_k,
\]
with $j$ maximal. Since $\mathcal{I} (1) = \emptyset$, such a $j$ exists. Let $A = C_k \setminus C_1$. Note that $A \in \mathcal{J}$ and
\[
 \mathcal{I}(1,\ldots,j-1) \subseteq A \subseteq C^* \setminus C_1 = \mathcal{I}(S_1^*).
\]
By the entry order condition, we know that $j$ will enter the active set before any variable in $N$, and before a perfect fit is reached. Thus $k+1\leq \tilde{T}$ and $C_{k+1}$ contains only additional interactions not involving any variables from $N$, so $C_{k+1} \subseteq C^*$. 
\qed

\vspace{\abovedisplayshortskip}
\paragraph{Proof of Corollary~\ref{cor:var_sel}.}
Let $\Omega_{C^*, \eta}$ be defined as in Lemma~\ref{lem:entry}. Also define the events 
\begin{gather*}
 \Omega_1 =  \{\tfrac{1}{n} \norms{\mb X_N ^T(\mb I - \mb P^{S^*}) \mbb\varepsilon}_\infty \leq \eta\}, \\
 \Omega_2 =  \{\tfrac{1}{n} \norms{\hat{\mbb\Sigma}_{S^*, S^*} ^{-1} \mb X_{S^*} ^T  \mbb\varepsilon}_\infty \leq \xi\}
\end{gather*}
In all that follows, we work on the event $\Omega_1 \cap \Omega_2 \cap \Omega_{C^*, \eta}$.
As $\mb I - \mb P^{S^*}$ is a projection,
\[
 \pr(\tfrac{1}{n}|{\mb X_v} ^T (\mb I - \mb P^{S^*}) \mbb\varepsilon| \leq \eta ) \geq \pr(\tfrac{1}{n}|{\mb X_v}^T \mbb\varepsilon| \leq \eta ).
\]
Further, $\tfrac{1}{n} \hat{\mbb\Sigma}_{S^*, S^*} ^{-1} \mb X_{S^*} ^T  \mbb\varepsilon \sim N_{|S^*|}(\mb 0,\tfrac{1}{n} \sigma^2 \hat{\mbb\Sigma}_{S^*, S^*} ^{-1} )$. Thus
\[
 \pr(\Omega_3) \geq |S^*| \pr (|Z| \leq \xi)
\]
where $Z \sim N(0,  \sigma^2 / (n c_{\mathrm{min}}(\hat{\mbb\Sigma}_{S^*, S^*})) )$.
Note that
\[
 \pr(\Omega_1 \cap \Omega_2 \cap \Omega_{C^*, \eta}) \geq 1 - \pr(\Omega_{C^*, \eta} ^c) - \pr(\Omega_1 ^c)- \pr(\Omega_2 ^c).
\]
Using this, it is straightforward to show that $\pr(\Omega_1 \cap \Omega_2 \cap \Omega_{C^*, \eta}) \geq 1 - 3\exp(-t^2 / 2)$. 

Since we are on $\Omega_{C^*, \eta}$, we can assume the existence of a $k^*$ from Theorem~\ref{thm:main}.
We now follow the proof of Lemma~\ref{lem:entry} taking $S = S^*$ and $M = C_{k^*} \setminus S^* \subseteq N$. The KKT conditions become
\begin{gather}
 \hat{\mbb\Sigma}_{S^*, S^*} (\mbb\beta^* _{S^*} - \hat{\mbb\beta}_{S^*}) + \tfrac{1}{n} \mb X_{S^*} ^T \mbb\varepsilon - \hat{\mbb\Sigma}_{S^*, M} \hat{\mbb\beta}_M = \lambda \hat{\mbb\tau}_{S^*}, \label{eq:KKT_21}\\
\hat{\mbb\Sigma}_{M, S^*} (\mbb\beta^* _{S^*} - \hat{\mbb\beta}_{S^*}) + \tfrac{1}{n} \mb X_M ^T \mbb\varepsilon - \hat{\mbb\Sigma}_{M, M} \hat{\mbb\beta}_M = \lambda \hat{\mbb\tau}_M, \label{eq:KKT_22}
\end{gather}
with $\hat{\mbb\tau}$ also satisfying \eqref{eq:infty_norm} and \eqref{eq:l1_norm_equiv} as before. Now let $\lambda$ be such that
\begin{equation*}
 \frac{\eta}{1 - \norms{\hat{\mbb\Sigma}_{M, S^*} \hat{\mbb\Sigma}_{S^*, S^*} ^{-1} \sgn(\mbb\beta^* _{S^*})}_\infty} < \lambda < \min_{v \in S^*} \left\{ \abs{\sgn(\mbb\beta^* _{S^*})^T (\hat{\mbb\Sigma}_{S^*, S^*} ^{-1})_v}^{-1} (| \beta^* _v | - \xi) \right\}.
\end{equation*}
It is straightforward to check that
\begin{gather*}
  (\hat{\mbb\beta}_{S^*}, \hat{\mbb\beta}_M) = (\mbb\beta^* _{S^*} - \lambda \hat{\mbb\Sigma}_{S^*, S^*} ^{-1} \sgn(\mbb\beta^* _{S^*}) + \tfrac{1}{n} \hat{\mbb\Sigma}_{S^*, S^*} ^{-1} \mb X_{S^*} ^T \mbb\varepsilon, \, \mb 0)  \\
 (\hat{\mbb\tau}_{S^*}, \hat{\mbb\tau}_M) = \left( \sgn(\mbb\beta^* _{S^*}), \, \hat{\mbb\Sigma}_{M, S^*} \hat{\mbb\Sigma}_{S^*, S^*} ^{-1} \sgn(\mbb\beta^* _{S^*}) + \tfrac{1}{n} \lambda^{-1} \mb X_M ^T (\mb I - \mb P^{S^*}) \mbb\varepsilon\right)
\end{gather*}
is the unique solution to \eqref{eq:KKT_21}, \eqref{eq:KKT_22}, \eqref{eq:infty_norm} and \eqref{eq:l1_norm_equiv}.
\qed

\section{Proof of Theorem~\ref{thm:random}}
In the following, we make use of notation defined in Section~\ref{sec:fixed}. In addition, for convenience we write $S=S_1^*$, $M= S \cup J^*$. Also, we will write main effects variables $\{j\}$ as simply $j$.
For any matrix $\mb M$, $\|\mb M\|_\infty$ will denote $\max_{jk} |M_{jk}|$. 
First we collect together various results concerning $\hat{\mbb\Sigma}_{C^*,C^*}$.
\begin{lem} \label{lem:hat_Sigma}
Consider the setup of Theorem~\ref{thm:random}. Let $\E_n$ and $\Var_n$ denote empirical expectation and variance with respect to $\mb Z$ so that, for example $\E_n z_j = \sum_{i=1}^n Z_{ij}/n$.
\begin{enumerate}[(i)]
\item Let $\mb D$ be the diagonal matrix indexed by $C^*$ used to scale transformations of $\mb Z$ in order to create $\mb X_{C^*}$ i.e.\ with entries such that $D_{jj}^2=\Var_n(z_j)$ and $D_{vv}^2 = \Var_n(z_j - \E_n z_j)(z_k - \E_n z_k)$ when $v=\{j,k\}$. Then
\begin{align}
\max_{j\in C_1} |D_{jj}^2 -1| &= O_P(\sqrt{\log(p)/n}) \label{eq:main_scale}\\
 \max_{\{j,k\} \in M} |D_{\{j,k\},\{j,k\}}^2 -1 -\Sigma_{jk}^2| &= O_P(\sqrt{\log(s_1)}n^{-1/4})  \label{eq:inter_scale}
\end{align}
\item 
\begin{align}
 \tfrac{1}{n}\|\mb X_{J^*} ^T \mb X_S\|_\infty &= O_P(\sqrt{\log(s_1)}n^{-1/3})  \label{eq:l_infty_inter} \\
  c_{\text{min}} (\hat{\mbb\Sigma}_{S,S}) &\geq c_* - s_1 O_P(\sqrt{\log(s_1)/n}) \\
 c_{\text{min}} (\hat{\mbb\Sigma}_{M,M}) &\geq c_*^2 +s_1^2O_P(\sqrt{\log(s_1)}n^{-1/4}) \\
 c_{\text{max}} (\hat{\mbb\Sigma}_{J^*,J^*}) &\leq 2{c^*}^2 + s_1^2O_P(\sqrt{\log(s_1)}n^{-1/4}). \label{eq:c_max_inter}
\end{align}
\end{enumerate}
\end{lem}
\begin{proof}
We use bounds on the tails of products of normal random variables from \citet{Hao2014} (equation B.9).
We have
\begin{align*}
\max_{j,k}|\Cov_n(z_j, z_k) -\Sigma_{jk}|&= \max_{j,k} |\E_n(z_jz_k) - \E_n z_j \E_n z_k - \Sigma_{jk}|\\
&= O_P(\sqrt{\log(p)/n}).
\end{align*}
Also,
\begin{align*}
&\max_{j,k,l,m \in S}|\Cov_n\big((z_j-\E_n z_j)(z_k - \E_n z_k),\, (z_l - \E_n z_l)(z_m - \E_n z_m)\big) - \Sigma_{jl}\Sigma_{km} - \Sigma_{jm}\Sigma_{kl}|
\\
&= \max_{j,k,l,m \in S}|\E_n(z_jz_kz_lz_m) -\E_n(z_jz_k)\E_n(z_lz_k) - \Sigma_{jl}\Sigma_{km} - \Sigma_{jm}\Sigma_{kl}| + O_P(\sqrt{\log(s_1)/n}) \\
&= O_P(\sqrt{\log(s_1)}n^{-1/4}).
\end{align*}

Now we consider (ii). We have
\begin{align*}
\tfrac{1}{n}\|\mb X_{J^*}^T \mb X_S\|_\infty &\leq \max_{v \in J^*} D_{vv}^{-1} \max_{k \in S}D_{kk}^{-1} \max_{j,k,l \in S} |\Cov_n\big((z_j-\E_nz_j)(z_k-\E_nz_k),\, z_l\big)| \\
&\leq O_P(\sqrt{\log(s_1)} n^{-1/3}),
\end{align*}
the rate being driven by the size of $\E_n(z_jz_kz_l)$.
Also
\begin{align*}
c_{\text{min}}(\hat{\mbb\Sigma}_{S,S}) &= \min_{\mbb\tau \in \R^{s_1}: \|\mbb\tau\|_2=1} \mbb\tau\{\mbb\Sigma_{S,S} - (\mbb\Sigma_{S,S} - \hat{\mbb\Sigma}_{S,S} )\}\mbb\tau \\
&\geq c_{\text{min}}(\mbb\Sigma_{S,S}) - \max_{\mbb\tau \in \R^{s_1}: \|\mbb\tau\|_2=1}\|\mbb\tau\|_1^2 \|\mbb\Sigma_{S,S} - \hat{\mbb\Sigma}_{S,S}\|_\infty \\
&= c_* - s_1 O_P(\sqrt{\log(s_1)/n}).
\end{align*}

Now let $\tilde{\mbb\Sigma}$ be a matrix with entries indexed by $M$ with
\[
\tilde{\Sigma}_{uv}=\Sigma_{jl}\Sigma_{km} + \Sigma_{jm}\Sigma_{kl}
\]
when $u=\{j,k\}$ and $v=\{l,m\}$. Lemma A.4 of \citet{Hao2014} shows that
$c_\text{min}(\tilde{\mbb\Sigma}) \geq 2c_{\text{min}}(\mbb\Sigma_{S,S})^2$ and $c_\text{max}(\tilde{\mbb\Sigma}) \leq 2 c_{\text{max}}(\mbb\Sigma_{S,S})^2$. Thus we have
\begin{align*}
c_{\text{min}}(\hat{\mbb\Sigma}_{M,M}) &= \min_{\mbb\tau \in \R^{|M|}:\|\mb D_{M,M} \mbb \tau\|_2=1} \mbb\tau \mb D_{M,M} \hat{\mbb\Sigma}_{M,M} \mb D_{M,M}\mbb\tau
\\
&\geq \|\mb D_{M,M}\|_\infty^{-1} c_{\text{min}}(\mb D_{M,M} \hat{\mbb\Sigma}_{M,M} \mb D_{M,M}) \\
&\geq \{1 + O_P(\sqrt{\log(s_1)}n^{-1/4})\} [c_*^2 -s_1^2 \{\|\tilde{\mbb\Sigma} - \mb D_{M,M} \hat{\mbb\Sigma}_{M,M} \mb D_{M,M})\|_\infty + O_P(\sqrt{\log(s_1)}n^{-1/3})\}] \\
&\geq c_*^2 +s_1^2O_P(\sqrt{\log(s_1)}n^{-1/4}). 
\end{align*}
Similarly
\begin{align*}
c_{\text{max}}(\hat{\mbb\Sigma}_{J^*,J^*}) &= \max_{\mbb\tau \in \R^{|J^*|}:\|\mb D_{J^*,J^*} \mbb \tau\|_2=1} \mbb\tau \mb D_{J^*,J^*} \hat{\mbb\Sigma}_{J^*,J^*} \mb D_{J^*,J^*}\mbb\tau
\\
&\leq \{1 - O_P(\sqrt{\log(s_1)}n^{-1/4})\} c_{\text{max}}(\mb D_{J^*,J^*} \hat{\mbb\Sigma}_{J^*,J^*} \mb D_{J^*,J^*}) \\
&\leq \{1 - O_P(\sqrt{\log(s_1)}n^{-1/4})\} \{2{c^*}^2 +s_1^2 \|\tilde{\mbb\Sigma} - \mb D_{J^*,J^*} \hat{\mbb\Sigma}_{J^*,J^*} \mb D_{J^*,J^*})\|_\infty\} \\
&\leq 2{c^*}^2 +s_1^2O_P(\sqrt{\log(s_1)}n^{-1/4}). 
\end{align*}
\end{proof}

\begin{lem}\label{lem:LHS_bd} Working with the assumptions of Theorem~\ref{thm:random}, we have
\[
\max_{A \in \mathcal{J}} \|\mbb\beta_S^{S \cup A} - \mbb\beta^*_S\|_\infty \leq O_P(\sqrt{s_1^3\log(s_1)}n^{-1/3}).
\]
\end{lem}
\begin{proof}
For $A \in \mathcal{J}$ let $\mbb\Delta^A \in \R^{|S \cup A|}$ with $\mbb\Delta^A_S = \mbb\beta^{S \cup A}_S - \mbb\beta^*_S$ and $\mbb\Delta^A_A=\mbb\beta^{S \cup A}_A$.
Define $\mb g^* = \mb X_{S_2^*} \mbb\beta^*_{S_2^*}$. Note that
\[
\mb f^* = \mb X_{S} \mbb\beta^*_{S} + \mb g^*,
\]
so
\[
\mbb\Delta^A =  (\mb X_{S \cup A}^T\mb X_{S \cup A})^{-1}\mb X_{S \cup A} \mb g^*.
\]
First we bound $\|\mbb\Delta^A_A\|_2^2$ in terms of $\|\mb g^*\|_2^2$.
We have that
\[
\|\mb X_{S\cup A}\mbb\Delta^A\|_2^2 = \|\mb X_S\mbb\Delta^A_S\|_2^2 + 2{\mbb\Delta^A_S}^T\mb X_S^T \mb X_A \mbb\Delta^A_A + \|\mb X_A \mbb\Delta^A_A\|_2^2 \leq \|\mb g^*\|_2^2.
\]
Thus
\[
c_{\text{min}}(\tfrac{1}{n}\mb X_S^T\mb X_S) \|\mbb\Delta^A_S\|_2^2 -2\sqrt{|A||S|} \|\tfrac{1}{n}\mb X_S^T \mb X_A\|_\infty \|\mbb\Delta^A_A\|_2\|\mbb\Delta^A_S\|_2 +  c_{\text{min}}(\tfrac{1}{n}\mb X_A^T\mb X_A) \|\mbb\Delta^A_A\|_2^2 - \tfrac{1}{n}\|\mb g^*\|_2^2 \leq 0.
\]
Thinking of this as a quadratic in $\|\mbb\Delta^A_S\|_2$ and considering the discriminant yields
\[
\|\mbb\Delta^A_A\|_2^2 \leq \frac{ \tfrac{1}{n} c_{\text{min}}(\tfrac{1}{n}\mb X_S^T\mb X_S)\|\mb g^*\|_2^2}{c_{\text{min}}(\tfrac{1}{n}\mb X_S^T\mb X_S) c_{\text{min}}(\tfrac{1}{n}\mb X_A^T\mb X_A) - \|\tfrac{1}{n}\mb X_S^T \mb X_A\|_\infty^2|A||S| }.
\]
Thus by Lemma~\ref{lem:hat_Sigma} (ii) and condition A2, $\max_{A \in \mathcal{J}}\|\mbb\Delta^A_A\|_2 = \tfrac{1}{\sqrt{n}}\|\mb g^*\|_2 O_P(1)$.

But
\[
\frac{1}{\sqrt{n}} \|\mb g^*\|_2 \leq \sqrt{c_{\text{max}}( \hat{\mbb\Sigma}_{J^*,J^*})} \|\mbb\beta^*_{S_2^*}\|_2 =O_P(1)
\]
by Lemma~\ref{lem:hat_Sigma} (ii) and A5, so $\max_{A \in \mathcal{J}}\|\mbb\Delta^A_A\|_2=O_P(1)$. 

Next observe that
\[
\|\mb X_{S \cup A}\mbb\Delta^A - \mb g^*\|_2^2 \leq \|\mb X_A \mbb\Delta^A_A - \mb g^*\|_2^2,
\]
so
\begin{align*}
\|\mbb\Delta^A_S\|_2^2 c_{\text{min}}(\tfrac{1}{n}\mb X_S^T\mb X_S) &\leq  \tfrac{1}{n} \|\mb X_S \mbb\Delta^A_S\|_2^2 \\
&\leq 2|\tfrac{1}{n}{\mbb\Delta^A_S}^T\mb X_S^T (\mb X_A \mbb\Delta^A_A - \mb g^*)| \\
&\leq 2\sqrt{|A||S|}\|\mbb\Delta^A_S\|_2 \|\tfrac{1}{n}\mb X_S^T \mb X_A\|_\infty \|\mbb\Delta^A_A\|_2 +2\|\mbb\Delta^A_S\|_2 \|\tfrac{1}{n}\mb X_S^T \mb g^*\|_2.
\end{align*}
Therefore
\begin{align*}
\|\mbb\Delta^A_S\|_\infty \leq 2\{c_{\text{min}}(\tfrac{1}{n}\mb X_S^T\mb X_S)\}^{-1} (\sqrt{|A||S|} \|\tfrac{1}{n}\mb X_S^T\mb X_A\|_\infty  \|\mbb\Delta^A_A\|_2 +\|\tfrac{1}{n} \mb X_S^T \mb g^*\|_2),
\end{align*}
so
\[
\max_{A \in \mathcal{J}} \|\mbb\Delta^A_S\|_\infty \leq 2\{c_{\text{min}}(\tfrac{1}{n}\mb X_S^T\mb X_S)\}^{-1} (\sqrt{|S||J^*|} \|\tfrac{1}{n}\mb X_S^T\mb X_{J^*}\|_\infty O_P(1) +\|\tfrac{1}{n} \mb X_S^T \mb g^*\|_2).
\]
Now
\begin{align*}
\|\tfrac{1}{n} \mb X_S^T \mb g^*\|_2 &\leq \sqrt{s_1}\|\tfrac{1}{n}\mb X_S^T\mb X_{S_2^*}\|_\infty \|\mbb\beta^*_{S_2^*}\|_1 \\
&\leq O_P(s_1 \sqrt{\log(s_1)}n^{-1/3}).
\end{align*}
Thus
\[
\max_{A \in \mathcal{J}} \|\mbb\Delta^A_S\|_\infty \leq O_P( \sqrt{s_1^3\log(s_1)}n^{-1/3}).
\]
\end{proof}

\vspace{\abovedisplayshortskip}
\paragraph{Proof of Theorem~\ref{thm:random}.}
In view of Theorem~\ref{thm:main} and its proof, it is enough to show that with probability tending to 1, we have
\begin{gather}
\max_{A \in \mathcal{J}} \sup_{\tau \in \R^{s_1}} \|\hat{\mbb\Sigma}_{N,S \cup A}\hat{\mbb\Sigma}^{-1}_{S \cup A,S \cup A}\mbb\tau\|_{\infty} < 1, \label{eq:rand_irep}\\
\min_{j \in I^*}\min_{A \in \mathcal{J}} |\beta_j^{S\cup A}| > \max_{A \in \mathcal{J}} \max_{j \in N} \bigg\{\frac{\tfrac{1}{n}|\mb X_j^T(\mb I -\mb P^{S\cup A})\mb f^*| + 2 \tfrac{1}{n}\|\mb X_{C^*}^T \mbb\varepsilon\|_\infty}{1 - \|\hat{\mbb\Sigma}^{-1}_{S \cup A,S \cup A}\hat{\mbb\Sigma}_{S \cup A,j}\|_{1}} + \tfrac{1}{n}\|\mb X_{C^*}^T \mbb\varepsilon\|_\infty\bigg\} \frac{\sqrt{|M|}}{c_{\text{min}}(\hat{\mbb\Sigma}_{M,M})}. \label{eq:rand_ent}
\end{gather}
First note that for $j \in N$, $\mb Z_j = \mb Z_S \mbb\Sigma_{S,S}^{-1} \mbb\Sigma_{S,j} + \mb E_j$ where $\mb E_j$ is independent of $\mb Z_S$ and $\mb E_j \sim N_n(\mb 0, (1 - \mbb\Sigma_{j,S}\mbb\Sigma_{S,S}^{-1}\mbb\Sigma_{S,j}) \mb I)$.
Thus
\[
 \mb X_j D_{jj} = \mb X_S \mb D_{S,S} \mbb\Sigma_{S,S}^{-1} \mbb\Sigma_{S,j} + \mb E_j - \mb{1}\bar{E}_j,
 \]
and
\begin{align*}
\max_{A \in \mathcal{J}} \|(\mb X_{S\cup A}^T \mb X_{S \cup A})^{-1} \mb X_{S\cup A}^T \mb X_j \|_1 &\leq D_{kk}^{-1} \| \mb D_{S,S} \mbb\Sigma_{S,S}^{-1} \mbb\Sigma_{S,j}\|_1 + \max_{A \in \mathcal{J}} \|\hat{\mbb\Sigma}_{S\cup A, S \cup A}^{-1} \tfrac{1}{n} \mb X_{S \cup A}^T \mb E_j\|_1.
\end{align*}
Now the second term above is at most
\begin{align*}
\max_{A \in \mathcal{J}} \max_{\mbb\tau \in \R^{|S \cup A|}:\|\mbb\tau\|_2 \leq 1} \|\hat{\mbb\Sigma}_{S\cup A,S \cup A}^{-1} \mbb\tau\|_1 \|\tfrac{1}{n}\mb X_{M}^T \mb E_j\|_2.
\end{align*}
But
\begin{align*}
\max_{A \in \mathcal{J}} \max_{\mbb\tau \in \R^{|S \cup A|}:\|\mbb\tau\|_\infty \leq 1} \|\hat{\mbb\Sigma}_{S\cup A,S \cup A}^{-1} \mbb\tau\|_1 &\leq \frac{\sqrt{|M|}}{c_{\text{min}}(\hat{\mbb\Sigma}_{M,M})} \\
&\leq \frac{\sqrt{|M|}}{c_*^2 + s_1^2O_P(\sqrt{\log(s_1)}n^{-1/4})}.
\end{align*}
Also since for $v \in M$ and $j \in N$, $\mb X_v^T\mb E_j /n \sim N(0,1)$ we have
\[
\max_{j \in N} \|\tfrac{1}{n}\mb X_{M}^T \mb E_j\|_2^2 \leq |M|O_P(\log(p)/n). 
\]
Therefore
\begin{align*}
\max_{A \in \mathcal{J}} \sup_{\tau \in \R^{s_1}} \|\hat{\mbb\Sigma}_{N,S \cup A}\hat{\mbb\Sigma}^{-1}_{S \cup A,S \cup A}\mbb\tau\|_{\infty} &\leq (1+o_P(1)) \delta + \frac{s_1^2 o_P(1)}{c_*^2 + o_P(1)}.
\end{align*}
This shows that \eqref{eq:rand_irep} is satisfied with probability tending to 1.

Next 
\begin{align*}
\max_{j \in N} \max_{A \in \mathcal{J}} \frac{1}{n}|\mb X_j^T(\mb I-\mb P^{S \cup A}) \mb f^*| = \max_{j \in N} \max_{A \in \mathcal{J}} \frac{D_{jj}^{-1}}{n} |\mb E_j^T(\mb I-\mb P^{S \cup A}) \mb X_{A} \mbb\beta^*_{A}|.
\end{align*}
Since $\mb E_j^T(\mb I-\mb P^{S \cup A}) \mb X_{A} \mbb\beta^*_{A}/n \sim N(0, \|(\mb I-\mb P^{S \cup A}) \mb X_{A} \mbb\beta^*_{A}\|_2^2/n^2)$ we have
\[
\max_{j \in N} \max_{A \in \mathcal{J}} \frac{1}{n}|\mb X_j^T(\mb I-\mb P^{S \cup A}) \mb f^*| \leq \sqrt{\frac{\log(2^{s_1} p)}{n}} \frac{1}{\sqrt{n}}\|\mb X_{S_2^*} \mbb\beta^*_{S_2^*}\|_2 O_P(1).
\]
By \eqref{eq:c_max_inter} we have 
\[
 \frac{1}{\sqrt{n}}\|\mb X_{S_2^*} \mbb\beta^*_{S_2^*}\|_2 \leq  \{2{c^*}^2 + s_1^2\sqrt{\log(s_1)}n^{-1/4} O_P(1)\} \|\mbb\beta_{S_2^*}\|_2.
\]
Now using Lemma~\ref{lem:LHS_bd} we see that the difference between the LHS and RHS of \eqref{eq:rand_ent} is at least
\begin{align*}
\min_{j \in I^*} |\beta^*_j| - O_P(\sqrt{s_1^3\log(s_1)}n^{-1/3}) - \left(\frac{(\sqrt{s_1 +\log p} + \sigma\sqrt{\log p})/\sqrt{n}}{1-\delta + o_P(1)} + \sigma\sqrt{\frac{\log(p)}{n}}\right)s_1 O_P(1).
\end{align*}
Thus A4 ensures that \eqref{eq:rand_ent} holds with probability tending to 1. 
\qed

\end{document}